\documentclass[sigconf, nonacm, authorversion,annonymous]{aamas} 

\usepackage{balance} 

\usepackage[utf8]{inputenc}
\usepackage{bm,xspace}
\usepackage{url} 
\usepackage{mathtools}
\usepackage{enumerate} 
\usepackage{tikz}
\usetikzlibrary{arrows,automata,calc,shapes,snakes,backgrounds,petri,mindmap,fit,positioning}
\definecolor{tucgreen}{RGB}{0,140,79}

\setcopyright{ifaamas}
\acmConference[AAMAS '22]{Proc.\@ of the 21st International Conference
on Autonomous Agents and Multiagent Systems (AAMAS 2022)}{May 9--13, 2022}
{Online}{P.~Faliszewski, V.~Mascardi, C.~Pelachaud,
M.E.~Taylor (eds.)}
\copyrightyear{2022}
\acmYear{2022}
\acmDOI{}
\acmPrice{}
\acmISBN{}

\acmSubmissionID{654}

\title[Reasoning about Human-Friendly Strategies in Repeated Keyword Auctions]{Reasoning about Human-Friendly Strategies in Repeated Keyword Auctions}

\author{Francesco	Belardinelli}
\affiliation{
  \institution{Université d’Evry}
  \city{Évry}
  \country{France}
}
\email{francesco.belardinelli@univ-evry.fr}

\author{Wojtek Jamroga}
\affiliation{
  \institution{University of Luxembourg}
  \city{Esch-sur-Alzette}
  \country{Luxemburg}
}
\email{wojciech.jamroga@uni.lu}

\author{Vadim Malvone}
\affiliation{
  \institution{Télécom Paris}
  \city{Paris}
  \country{France}
}
\email{vadim.malvone@telecom-paris.fr}

\author{Munyque Mittelmann} 
\affiliation{
  \institution{Université de Toulouse - IRIT}
  \city{Toulouse}
  \country{France}
}
\email{munyque.mittelmann@irit.fr}

\author{Aniello Murano }
\affiliation{
  \institution{University of Naples Federico II}
  \city{Naples}
  \country{Italy}
  }
  \email{nello.murano@gmail.com}
  
\author{Laurent Perrussel}
\affiliation{
  \institution{Université de Toulouse - IRIT}
  \city{Toulouse}
  \country{France}
  } 
 \email{laurent.perrussel@irit.fr}

\begin{abstract}

In online advertising, search engines sell ad placements for keywords continuously through auctions. This problem can be seen as an infinitely repeated game since the auction is executed whenever a user performs a query with the keyword. As advertisers may frequently change their bids, the game will have a large set of equilibria with potentially complex strategies. In this paper, we propose the use of natural strategies for reasoning in such setting as they are processable by artificial agents with limited memory and/or computational power as well as understandable by human users.  To reach this goal, we introduce a quantitative version of Strategy Logic with natural strategies in the setting of imperfect information. In a first step, we show how to model strategies for repeated keyword auctions and take advantage of the model for proving properties evaluating this game. In a second step, we study the logic in relation to the distinguishing power, expressivity, and model-checking complexity for strategies with and without recall. 
\end{abstract}


\usepackage{macros,environments}

\begin{document}


\pagestyle{fancy}
\fancyhead{}


\maketitle 

\section{Introduction}
\label{sec:introduction}

In recent years a wealth of logic-based languages have been introduced to reason about the strategic abilities of autonomous agents in multi-agent systems (MAS), including Alternating-time Temporal Logic (\ATL) \cite{alur2002alternating}, Strategy Logic (\SL) \cite{MMPV14,chatterjee2010strategy}, and Game Logic \cite{pauly2003game}, just to name a few. 
In conjunction with model checking techniques  \cite{principlesmodelchecking}, these formal languages have allowed for the development of efficient verification tools \cite{LomuscioQuRaimondi15,GammieMeyden04a,Cermak+14}, which have been successfully applied to the certification of MAS as different as voting protocols \cite{BelardinelliCDJ17,JamrogaKM20}, robot swarms \cite{Dixon+12,KouvarosLomuscio15}, and business processes \cite{Deutsch+09,GonzalezGriesmayerLomuscio15}.

Still, verification tools and techniques are comparatively less developed for {\em data-driven} and {\em data-intensive} systems\footnote{``[Model checking] is mainly appropriate to control-intensive applications and less suited for data-intensive applications" \cite[p.~15]{principlesmodelchecking}}, that is, contexts where the data content of processes, or {\em agents}, is key to model and account for the evolution of the system \cite{BelardinelliLomuscioPatrizi15,MontaliCalvaneseGiacomo14}.
This is the case also for online advertising, where search engines sell ad placements for keywords continuously through auctions. This problem can be seen as an infinitely repeated game since the auction is executed whenever a user performs a query with the keyword. As advertisers may frequently change their bids, the game will have a large set of equilibria with potentially complex strategies, thus making the specification and verification of keyword auctions a complex problem to solve for current model checking methods\footnote{``In principle, the sets of equilibria in such repeated games can be very large, with players potentially punishing each other for deviations. The strategies required to support such equilibria are usually quite complex, however, requiring precise knowledge of the environment and careful implementation. In theory, advertisers could implement such strategies via automated robots, but in practice they may not be able to: bidding software must first be authorized by the search engines, and search engines are unlikely to permit strategies that would allow advertisers to collude and substantially reduce revenues." \cite{edelman2007internet}}.

In this paper, we propose the use of \emph{natural strategies}~\cite{natStrategy,natStrategyII} for reasoning about 
equilibria in keyword auctions.
Strategies in MAS are typically defined as functions from (sequences of) states to actions. 
The authors of \cite{natStrategy,natStrategyII} argued that such ``combinatorial'' strategies might be appropriate to model the strategic abilities of a machine (robot, computer program) with extensive computational power. However, they do not provide good models of behaviour for agents with limited memory and computing capacity, such as humans. 
As a remedy, they proposed to model ``human-friendly'' strategies by lists of condition-action pairs with bounded complexity. This is consistent with the empirical research on {human concept learning}~\cite{Bourne70concepts,Feldman00conceptlearning} and social norms~\cite{Santos18SocialNormComplexity,Santos18phd}, as well as some works on usability~\cite{Nielsen94usability} and psychology of planning~\cite{Morris14psycho-planning}.
Natural strategies have been already used to redefine some security requirements for voting protocols in~\cite{JamrogaKM20}.

In our case, the bidding strategy in an auction should be executable for a \emph{simple} artificial agent, as well as reasonably transparent to the human user, which makes natural strategies a good match. 
Moreover, natural strategies provide a way to define complexity 
(and hence also ``simplicity'') 
metrics for various functionality, security, and usability properties in MAS. By focusing on simple strategies, one can make the 
verification of equilibrium properties decidable, or even tractable, despite the prohibitive complexity of the general problem. This is especially evident for strategies with memory, which normally make the synthesis and model checking problems undecidable~\cite{Dima11undecidable,Vester13ATL-finite}. 
%

\head{Contribution} 
By leveraging on natural strategies, we introduce a quantitative semantics for \SL\ with natural strategies and  imperfect information. As a first contribution, we show how to represent popular strategies for repeated keyword auctions in the proposed framework, as well as prove properties pertaining to this game.
Second, we analyse our novel variant of \SL\ in relation with its distinguishing power, expressivity, and  complexity of the model checking problem, for natural strategies with and without recall.
%
%
%
\subsection{Related work} 
Recently, there have been efforts to apply 
formal methods to the (semi-)automatic 
verification of some 
decision-making mechanisms, including auctions and voting protocols. 
A number of works~\cite{Caminati2015,barthe2016computer,KERBER201626} expresses these mechanisms in high-level specification languages. However, in contrast with standard model checking techniques \cite{principlesmodelchecking}, their verification is not fully automated, but only assisted by a reasoner. 
\citet{THW11} 
introduce a framework for fully-automated verification of voting protocols. Still, their approach can only model one-shot mechanisms and thus does not capture multi-stage protocols and repeated auctions. 
In preliminary works, 
\citet{pauly2003logic} and \citet{wool2007} 
advocate the use of 
\ATL\ \cite{alur2002alternating} to reason about decision-making mechanisms.
As \ATL\ lacks the expressivity to reason about quantitative aspects such as valuations and payments, and solution concepts such as equilibria, \citet{KR2021} introduce \SLKF, a quantitative and epistemic version of \SL\ 
\cite{chatterjee2010strategy,MMPV14}, and show how it can be used for reasoning about notions such as Nash equilibrium 
and strategyproofness. 
Still, their approach considers strategies as functions from states to actions and cannot handle strategies with recall. 
 

A key assumption of the present contribution is that agents have only partial observability of the global state of the system, as it is often the case in real-life applications.
Contexts of imperfect information have been extensively considered in the literature on formal verification~\cite{Dima11undecidable,KV00,JA06,Reif84,BJ14}.
Generally speaking, imperfect information immediately entails higher complexity of game solving.
In multi-player games, the complexity can go up to being non-elementary~\cite{PR89}, or even undecidability when considered in the context of memoryful strategies~\cite{Dima11undecidable}. Hence, it is of interest to analyse imperfect information systems where agents have finite or bounded memory, in order to retrieve a decidable model checking problem.
%
%
Works that are closest in spirit to our contribution concern modeling, specification, and reasoning about strategies of bounded-memory agents.
We directly build on the research by Jamroga, Malvone, and Murano on natural strategies~\cite{natStrategy,natStrategyII}. 
We generalize the approach by considering quantitative semantics for both natural strategies and the logic, which is more suitable for  reasoning about mechanisms with monetary transfer (\eg, auctions). 
We also consider \SL\ instead of \ATL, due to its expressive power. 
In a related vein, \citet{Agotnes09bounded} investigate strategic abilities of agents with bounded memory, while~\citet{BLM18} consider bounded memory as an approximation of perfect recall. On a related direction, temporal and strategic logics have been extended to handle agents with bounded resources \cite{Alechina09bounded,Alechina10atl-bounded,Bulling10rtl,Bulling10bounded}. Issues related to bounded rationality are also investigated in~\cite{BCS08,HO09,GSW14}.

Also relevant for the present contribution are papers that study explicit representations of strategies.
This category is much richer and includes extensions of $\ATLs$ with explicit reasoning about actions and strategies~\cite{Hoek05commitment,Agotnes06action,Walther07explicit,Herzig14explicit-progs}, as well as logics that combine features of temporal and dynamic logic~\cite{Harel82processlogic,Novak09codepatterns}. 
\citet{Duijf16strategies} present a variant of STIT logic, that enables reasoning about strategies and their performance in the object language.
Also, plans in agent-oriented programming are in fact rule-based descriptions of strategies. In particular, reasoning about agent programs using strategic logics was investigated in~\cite{Bordini06verifying,Alechina07aprograms,Alechina08execution,Dastani10programs-aamas,Yadav12atl-like}.




\head{Outline} 
In Sec.~\ref{sec:preliminaries}, we recall basic definitions. 
In Sec.~\ref{sec:natslf}, we define Natural Strategy Logic, 
denoted \NatSLF. 
In Sec.~\ref{sec:keywordauction} we focus on the problem of repeated keyword auctions. 
In Sec.~\ref{sec:expressivity}, we investigate the expressivity and distinguishing power of \NatSLF.
Sec.~\ref{sec:modelchecking} establishes the complexity of model checking, and  Sec.~\ref{sec:conclusion} concludes the paper. 
\section{Preliminaries}
\label{sec:preliminaries}

We first recall basic notions. 
For the remainder of the paper, we fix a set of atomic propositions $\APf$, a 
set of agents $\Ag$ and a set of strategy variables $\Var$. 
We let $ \qtdAg$ be the number of agents in $\Ag$. 
Finally, let $\Func\subseteq \{f\colon [\lowb,1]^m\to [\lowb,1] \mid m \in \setn \}$  be a set of functions over $[\lowb,1]$ of possibly different arities. 
 
\subsection{Weighted Concurrent Game Structures}

The semantics of natural strategies and \NatSLF\  are interpreted over weighted concurrent game structures (\wCGS). 
A difference from classical structures is that the labelling of atomic propositions is replaced by a weight function. We consider weighted propositions for easily handling quantitative aspects (such as prices). 

\begin{definition}
  \label{def-wcgs}
A \emph{weighted concurrent game structure with imperfect information} (\wCGS) is a tuple
$\wCGS=(\Act,\setpos,\setlegal, 
\trans,\val,\setpos_\init,\{\obsrel\}_{\ag\in\Ag})$ where: 
(i)    $\Act$ is a finite set of \emph{actions};
(ii)    $\setpos$ is a finite set of \emph{states}; 
(iii)    $\setlegal: \Ag\times\setpos\to 2^{\Act}$ is a \emph{legality function}, defining the availability of actions; 
(iv)    $\trans$ is a transition function assigning a successor state $\pos' = \trans(\pos, (\act_\ag)_{\ag \in \Ag})$ to each state $\pos \in \setpos$ and any tuple of actions $(\act_\ag)_{\ag \in \Ag}$, where $\act_\ag \in \setlegal(\ag, \pos)$; 
(v)    $\val:\setpos\times\APf\to [\lowb,1]$ is a \emph{weight function};
(vi)   $\setpos_\init\subseteq\setpos$ is a set of \emph{initial states}; 
    and 
(vii)    $\obsrel\;\subseteq \setpos\times\setpos$ is an equivalence relation called the  \emph{observation relation} of agent $\ag$. 
\end{definition} 
We require that the \wCGS is uniform, that is $\pos \obsrel \pos'$ implies $\setlegal(\ag, \pos) = \setlegal(\ag, \pos')$. 
We  write $\profile{o}$ for a tuple of objects $(o_\ag)_{\ag\in\Ag}$, one for each agent, and such tuples are called \emph{profiles}.
Given a profile $\profile{o}$ and $\ag\in\Ag$, we let $o_\ag$ be agent $\ag$'s component, and $o_{-\ag}$ is $(o_\agb)_{\agb \in \Ag\setminus\{\ag\}}$. Similarly, we let $\Ag_{-\ag}=\Ag\setminus\{\ag\}$.

In a state $\pos\in\setpos$, each player $\ag$ chooses an available action $\mova\in \setlegal(\ag,\pos)$, 
and the game proceeds to state
$\trans(\pos, \jmov)$ where $\jmov$ is the \emph{action profile}
$(\mova)_{\ag\in\Ag}$. 
%
A \emph{play} $\iplay=\pos_{0}\pos_1\pos_2\ldots$ is an infinite
sequence of states  
such that for every $i\geq 0$ there exists an action profile $\jmov$
such that $\trans(\pos_{i}, \jmov)=\pos_{i+1}$. 
We write $\iplay_i=\pos_i$ for the state at index $i$ in play $\iplay$. 
A \textit{history} $\history = \pos_{0}\pos_1\pos_2\ldots\pos_n$ is a finite sequence of states. The last element of a history is denoted by $\lasth(\history) = \pos_n$. 
$\sethistory$ denotes the set of all histories in the wCGS $\wCGS$.

\subsection{Natural Strategies}

In this section we recall the notion of uniform natural strategies from \cite{natStrategyII}. 
Natural strategies are conditional plans, represented through an ordered list of condition-action rules \cite{natStrategyII}. The intuition is that the first rule whose condition holds in the history of the game is selected, and the corresponding action is executed. 
As we are considering the setting of imperfect information, the conditions are regular expressions over \textit{weighted epistemic  (WE) formulas}. 
Given an agent $\ag$, the WE formulas over $\Prop$, denoted $\WE(\Prop)$, are conditions on $\ag$'s knowledge and are expressed by the following Backus-Naur Form grammar: 
\[\psi \coloncolonequals \top \mid \Ka \phi \mid f(\psi, \dots, \psi) 
\] 
\[\phi \coloncolonequals p \mid f(\phi, \dots, \phi)  \mid \Kb\phi \]  
where  $f\in \Func$ is a function, $p \in \Prop$ is an atomic proposition and $\agb \in \Ag$ is an agent. 

Given a wCGS $\wCGS$, a state $\pos \in \setpos$ and a $\WE(\Prop)$ formula $\phi$, we inductively define the satisfaction value of $\phi$ in $\pos$, denoted $\semEP{\pos}{\phi}$: 
\begingroup  \allowdisplaybreaks
  \begin{align*}
    \semEP{\pos}{p} &= \val(\pos,p) \\
 \semEP{\pos}{\Ka\phi} &=
    \min_{\pos'\obsrel\pos}\semEP{\pos'}{\phi} \\
    \semEP{\pos}{f(\varphi_{1}, \!\ldots, \varphi_{m})} & =
    f(\semEP{\pos}{\varphi_{1}}, \ldots,
    \semEP{\pos}{\varphi_{m}}) 
  \end{align*}
  \endgroup

The semantics for the knowledge modality is the standard in the literature on 
fuzzy epistemic logic (\eg\   \cite{maruyama2021reasoning}). 
Let $\setregular(\WE(\Prop))$ be the set of regular expressions over the weighted epistemic conditions $\WE(\Prop)$, defined with the constructors $\concat, \ndchoice, \iteration$ representing concatenation, nondeterministic choice, and finite iteration, respectively. 
Given a regular expression $\regular$ and the language $\Language(\regular)$ on words generated by $\regular$, a history $\history$ is \textit{consistent} with $\regular$ iff there exists $b \in \Language(\regular)$ such that $|\history| = |b|$ and $\semEP{\history[i]}{b[i]}=1$, for all $0 \leq i \leq |\history|$.
Intuitively, a history $\history$ is consistent with a regular expression $\regular$ if the $i$-th weighted epistemic condition in $\regular$ ``holds'' in the $i$-th state of $\history$ (for any position $i$ in $\history$). 

A \textit{uniform natural strategy with recall} $\strat_\ag$ for agent $\ag$ is a sequence of pairs $(\regular, \act)$, where $\regular \in \setregular(\WE(\Prop))$ is a regular expression, and $\act$ is an action available in $\last(\history)$, for all histories $\history \in \sethistory$ consistent with $\regular$. The last pair on the sequence is required to be  $(\top\iteration, \act)$, with $\act \in \setlegal(\ag, \pos)$ for every $\pos \in \setpos$ and some $\act \in \Act$.  

A \textit{uniform memoryless natural strategy} is a special case of natural strategy in which each condition is a weighted epistemic formula (i.e., no regular operators are allowed).

Natural strategies are uniform in the sense they specify the same actions in indistinguishable states (see \cite{natStrategyII}).  
We define $\setstrata^{\setting}$ to be the set of uniform 
natural strategies for agent $\ag$ and  $\setstrat^{\setting} = \cup_{\ag \in \Ag} \setstrata^{\setting}$, where $\setting \in \{ir, iR\}$\footnote{As usual in the verification process, we denote imperfect recall with r, perfect recall with R, imperfect information with i, and perfect information with I.}. 
Let $\size(\strat_\ag)$ denote the number of guarded actions in $\strat_\ag$, $\cond_\idx(\strat_\ag)$ be the $\idx$-th guarded condition on $\strat_\ag$, $\cond_\idx(\strat_\ag)[j]$ be the $j$-th WE formula of the guarded condition $\strat_\ag$, and $\action_\idx(\strat_\ag)$ be the corresponding action. Finally, $\match(\history, \strat_\ag)$ is the smallest index $\idx \leq \size(\strat_\ag)$ such that for all $0 \leq j \leq |last(h)|$, $\semEP{\history[j]}{\cond_\idx(\strat_\ag)[j]} = 1$\footnote{Note that, we considered the case in which the condition have the same length of the history. There is also the case in which the condition is shorter than the history. This is due to the usage of the finite iteration operator. In the latter case, we need to check a finite number of times the same weighted epistemic formula in different states of the history. For more details on this aspect see \cite{natStrategy,natStrategyII}.} and $\action_\idx(\strat_\ag) \in \setlegal(\ag, \lasth(\history))$. In other words, $\match(\history, \strat_\ag)$ matches the state $\lasth(\history)$ with the first condition in $\strat_\ag$ that holds in $\history$, and action available in $\lasth(\history)$. 

\halfline
\head{Measurement of Natural Strategies.}
The complexity of the strategy $\strat$ is the total size of its representation and is denoted as follows:  
$compl(\strat) \egdef \sum_{(\regular, \act) \in \strat}|\regular|$, 
where $|\regular|$ is the number of symbols in $\regular$, except by parentheses. If  $\regular$ is a $n$-ary function in $\Func$, then  $|\regular| = n+1$.  


\section{Natural Strategy Logic}
\label{sec:natslf}

\SLF\ \cite{Bouyer19} proposes a quantitative semantics for Strategy Logic, in which strategies are functions mapping  histories to actions. 
For reasoning about intuitive and simple  strategies, we introduce \SLF\ with natural strategies and imperfect information, denoted \NatSLF. 
Throughout this section, let $\setting \in \{ir, iR\}$ denote whether the semantics considers memoryless or recall strategies. 

An \emph{assignment}  $\assign:\Ag\union\Var \to \setstrat^{\setting}$ is
a function from players and variables to strategies.
For an assignment
$\assign$, an agent $\ag$ and a strategy $\strat$ for $\ag$,
$\assign[a\mapsto\strat]$ is the assignment that maps $a$ to $\strat$ and is otherwise equal to
$\assign$, and 
$\assign[\var\mapsto \strat]$ is defined similarly, where $\var$ is a
variable. 
%
For an assignment $\assign$ and a state $\pos$ 
we let
$\out(\assign,\pos)$ be the unique play that starts in  
$\pos$ and follows the strategies
assigned by $\assign$.  
Formally,  $\out(\assign,\pos)$ is the play $\pos_{0}\pos_{1}\ldots$ such that $\pos_0=\pos$ and for all $i\geq 0$, $\pos_{i+1}=\trans(\pos_{i},\jmov)$ where for all $\ag\in\Ag$,  $\jmov_\ag = \action_{\match(\pos_i, \assign(\ag))}(\assign(\ag))$.  

\subsection{\NatSLF\ Syntax}
\begin{definition}
  The syntax of \NatSLF\  is defined as follows: 
\begin{align*}
  	\varphi  ::= p \mid \Estrata{\ag}\varphi \mid (\ag, \var_\ag)  \varphi  
        \mid f(\varphi, \ldots, \varphi)
        \mid \X \phi \mid \phi \Uu \phi
\end{align*}
where $p\in\APf$,  $\var_\ag\in\Var \ \cup \ \setstrat^\setting_\ag$, $\ag\in\Ag$, and $f\in \Func$.
\end{definition}

The intuitive reading of the operators is as follows: $\Estrata{\ag}\varphi$ means that there exists a strategy with complexity less or equal than $k$ for agent $\ag$ such that $\phi$ holds; $(\ag,\var_\ag)\varphi$ means that when strategy $\var_\ag$ is assigned to agent $\ag$, $\phi$ holds; 
$\X$ and $\Uu$ are the usual temporal operators ``next'' and ``until''. The meaning of $f(\varphi_1,\ldots,\varphi_n)$ depends on the function $f$. 
We use $\top$, $\ou$, and $\neg$ to denote, respectively, function~$1$, function $x,y\mapsto \max(x,y)$ and
function $x\mapsto - x$.

A variable is \emph{free} in formula $\phi$ if it is bound to an agent without being quantified upon, and an agent $\ag$ is free in $\phi$ if $\phi$ contains a temporal operator ($\X$ or $\Uu$) that is not in the scope of any binding for $\ag$. The set of free variables and agents in $\phi$ is written $\free(\phi)$, and a formula $\phi$ is a \emph{sentence} if $\free(\phi)=\emptyset$. 
The strategy quantifier $\Estrata{\ag}\phi$ quantifies on strategies \emph{for agent $\ag$}.  

\subsection{\NatSLF\ Semantics}

\begin{definition}
Let $\wCGS=(\Act,\setpos,\trans,\val,\setpos_\init, \{\obsrel\}_{\ag\in\Ag})$ be a \wCGS, and $\assign$ an assignment. The satisfaction value $\semSL{\CGS}{\setting}{\assign}{\pos}{\varphi}\in[\lowb,1]$ of  a \NatSLF\ formula $\phi$ in a
  state~$\pos$
  is defined as follows, 
  where $\iplay$ denotes $\out(\pos,\assign)$:
  \begingroup
  \allowdisplaybreaks
  \begin{align*}
  \semSL{\CGS}{\setting}{\assign}{\pos}{p} &= \val(\pos,p) \\
      \semSL{\CGS}{\setting}{\assign}{\pos}{\Estrata{\ag}\varphi} &=  \max_{\strat \in \{\alpha \in \setstrata^{\setting}        : compl(\alpha) \leq k\} }
    \semSL{\CGS}{\setting}{\assign[\var_\ag \mapsto \strat]}{\pos}{\varphi} \\ 
  \semSL{\CGS}{\setting}{\assign}{\pos}{(\ag,\var_\ag)\varphi} & = \semSL{\CGS}{\setting}{\assign[\ag \mapsto
      \assign(\var_\ag)]}{\pos}{\varphi} \text{ if }\var_\ag \in \Var\\
 \semSL{\CGS}{\setting}{\assign}{\pos}{(\ag,\strat_\ag)\varphi} & = \semSL{\CGS}{\setting}{\assign[\ag \mapsto
      \strat_\ag]}{\pos}{\varphi} \text{ if }\strat_\ag \not \in \Var\\
  \semSL{\CGS}{\setting}{\assign}{\pos}{f(\varphi_{1}, \!\ldots, \varphi_{m})} & =
    f(\semSL{\CGS}{\setting}{\assign}{\pos}{\varphi_{1}}, \ldots,
  \semSL{\CGS}{\setting}{\assign}{\pos}{\varphi_{m}}) \\
  \semSL{\CGS}{\setting}{\assign}{\pos}{\X \phi} &= \semSL{\CGS}{\setting}{\assign}{\iplay_{1}}{\phi} \\
  \semSL{\CGS}{\setting}{\assign}{\pos}{\phi_1\Uu\phi_2} & = \sup_{i \ge 0}     \min
       \Big
      (
    \semSL{\CGS}{\setting}{\assign}{\iplay_{i}}{\phi_2},\! 
      \min_{0\leq j <i} 
    \semSL{\CGS}{\setting}{\assign}{\iplay_{j}}{\phi_1}   \Big)
  \end{align*}
  \endgroup
\end{definition}

If $\phi$ is a sentence, its satisfaction value does not depend on the assignment, and we write $\semSL{\CGS}{\setting}{}{\pos}{\phi}$ for $\semSL{\CGS}{\setting}{\assign}{\pos}{\phi}$ where $\assign$ is any assignment. We also let $\semSLglobal{\CGS}{\setting}{\phi}=\min_{\pos_\init\in\setpos_\init}\semSL{\CGS}{\setting}{}{\pos_\init}{\phi}$.

\begin{remark} When propositions only take values in $\{-1, 1\}$ and $\Func=\{\top,\lor,\neg\}$,  \NatSLF\ corresponds to a Boolean-valuated extension of \SL\ with Natural Strategies. 
\end{remark}

We define the  
classic abbreviations: ${\perp\colonequals\neg\top}$,
${\varphi\rightarrow \varphi' \colonequals \neg
\varphi \vee \varphi'}$, 
${\varphi\wedge\varphi'
\colonequals \neg (\neg \varphi \vee \neg \varphi')}$,
${\F\psi \colonequals \top \Uu \psi}$,
${\Gg\psi \colonequals \neg \F \neg \psi}$ and  $\Astrat\varphi\colonequals$ $\neg\Estrat\neg\varphi$,
and check that they correspond to the intuition. For instance, $\wedge$ corresponds to $\min$, $\F\psi$ computes the supremum 
of the satisfaction value of $\psi$ over all future points in time, $\Gg\psi$ computes the 
infimum of these values, and $\Astrat\varphi$ minimizes the value of $\varphi$ over all possible strategies $\var$. 

\section{Repeated Keyword Auctions} 
\label{sec:keywordauction}

Modeling mechanisms with monetary transfer and private valuations require handling quantitative features and  imperfect information. 
Memoryless strategies are enough for mechanisms in which all relevant information is encoded in the current state (\eg\ English auction). In repeated auctions, agents may, as well, use information from the previous states for choosing their strategies.

We now focus on using \NatSLF\ to model and verify repeated keyword auctions and related strategies. 
Repeated keyword auctions are used by online search engines 
for selling advertising slots when users perform a search with a keyword \cite{Cary2007}. 
For a keyword of interest, the advertisers (bidders)  submit a bid stating the maximum amount she is willing to pay for a click on her sponsored link. 
When a user submits a query, an auction is run to determinate the slot allocation among the advertisers bidding  on the keyword of interest. 
The most common mechanism for keyword auctions is the Generalized Second Price (GSP) \cite{Cary2007}, in which the agents are allocated slots in decreasing order of bids and the payment for the slot $\slot$ is the bid of the agent allocated to the slot $\slot+1$.  

We assume that $\Func$ contains the function $\pref\,:(x,y)\mapsto 
1$ if $x\leq y$ and $\pref\,:(x,y)\mapsto\lowb$ otherwise; 
and for readability we use the infix notation $x \pref y$ in the formula.  We also assume that $\Func$ contains the equality $=$ and comparison functions $<$, $>$, $\geq$
(defined similarly). 
Finally, we assume $\Func$ contains functions $-$, $\sum$, $\times$, $\backslash$, $min$, $max$ and $argmax$ with the standard meaning (for details, see  \cite{KR2021}). 

Let us fix a price increment $\increment \in (0,1]$, a set of slots $\Slot = \{1, \dots, \qtdSlot\}$, where $\qtdSlot \in \mathbb{N}\setminus\{0\}$. 
Each slot has a click-through rate $\theta_1 > \ldots > \theta_\qtdSlot$, where $\theta_\slot \in [0,1]$ is the probability that the user will click on the advertisement in slot $\slot$. 
The agents in $\Ag$ are the advertisers, each one having a private valuation $\valuation_\ag \in \Valuation_\ag
$ for a click, where $\Valuation_\ag \subset [0,1]$ is a finite set of possible valuations. We assume the valuations are distinct, that is, if $\ag \neq \ag'$, then $\valuation_\ag \neq \valuation_{\ag'}$. 
We denote by $\prec$ an arbitrary order among the agents in $\Ag$, used in case of ties.  
The atomic propositional set is $\Prop = \{\allocprop_{\ag, \slot}, \paymentprop_{\slot}, \valprop_\ag: \ag \in \Ag, \slot \in \Slot\}$, where $\allocprop_{\ag, \slot}$ represents whether agent $\ag$ is allocated to slot $\slot$,  $\paymentprop_{\slot}$ denotes the price of slot $\slot$ and $\valprop_\ag$ denotes $\ag$'s valuation. 
Define $\wCGS_{GSP}=(\Act 
,\setpos,\setlegal,\trans,\val,\setpos_\init,\{\obsrel\}_{\ag\in\Ag})$, where: 
\begin{itemize}
    \item $\Act = \{0+x \times \increment : 0 \leq x \leq \frac{1}{\increment}\}$, where $b \in \Act$ denotes a bid with price $b$ for a click; 
    given 
    $\jmov = (\act_\ag)_{\ag \in \Ag}$, let $rank_{\jmov} = (\ag_1, \dots, \ag_\qtdAg)$ be the sequence of distinct agents in $\Ag$ ordered by their bid, that is,   $i<j$ if $  \act_{\ag_i} > \act_{\ag_j} $ or $ \act_{\ag_i} = \act_{\ag_j} $ and $ \ag_i \prec \ag_j$
    for  $i, j \in \{1,\dots,\qtdAg\}$ with $i \neq j$. In case of draws, the sequence is determined with respect to $\prec$. We let $rank_{\jmov}(i)$ denote the agent in the $i$-th position of the sequence $rank_{\jmov}$.

    \item $\setpos= \{ \statetuple : al_\slot \in \Ag\cup\{none\}$ \& $pr_\slot \in \Act$ \& $ vl_\ag \in \Valuation_\ag$ \& $ \ag \in \Ag$ \& $ 1\leq\slot\leq\qtdSlot\}$, where each state represents the current slot allocation and prices, with  $al_\slot$, $pr_\slot$, and $vl_\ag$ denoting the winner of slot  $\slot$, the price per click of $\slot$ and $\ag$'s valuation, resp.; 
    
    \item For each $\ag \in \Ag$ and $\pos \in \setpos$, $\setlegal(\ag, \pos) = \Act$;
    
    \item For each $\pos \in \setpos$ and 
    $\jmov = (\act_\ag)_{\ag \in \Ag}$ such that $\act_\ag \in \setlegal(\ag, \pos)$, the transition function uses the agent's bids to chose the next allocations and prices and is defined as follows:  $\trans(\pos, (\act_\ag)_{\ag \in \Ag}) = \statetupleprimeval$, where for each agent $\ag$ and slot $\slot$,  (i) $al_{\slot} = 
       rank_{\jmov}(\slot)$ if $ \slot \leq \qtdAg$, and $al_{\slot} =
        none $ otherwise; (ii) $pr_{\slot} = \act_{rank_{\jmov}(\slot+1)} $ if $\slot+1\leq\qtdAg$, and $pr_{\slot} = 0$ otherwise. 

    \item  For each agent $\ag$, slot $\slot \in \Slot$ and state $\pos = \statetuple$, the weight function is defined as follows:
        (i) $\val(\pos, \allocprop_{\ag, \slot}) = 1$ if $al_{\slot}= \ag$, and $\val(\pos, \allocprop_{\ag, \slot}) = 0$ otherwise;
        (ii) $\val(\pos, \paymentprop_{\slot}) = pr_{\slot}$; and
        (iii) $\val(\pos, \valprop_{\ag}) = vl_{\ag}$. 
    
    \item In an initial state, the prices are 0 and the slots are allocated to $none$, that is,     $\setpos_\init = \{\langle none,\allowbreak \ldots,\allowbreak none,\allowbreak 0,\allowbreak \ldots,\allowbreak 0,\allowbreak  vl_{1},\allowbreak \ldots,\allowbreak vl_{\qtdAg} \rangle  \in \setpos\}$; 
    
    \item For each agent $\ag$ and two states $\pos = \statetuple$ and $\pos' = \statetupleprime$ in $\setpos$, the observation relation 
    $\obsrel$ is such that  
    if $\pos \obsrel \pos'$ then 
    (i) $al_{\slot} = al_{\slot'}$, for each $1 \leq \slot \leq \qtdSlot$;
    (ii) $p_{\slot} = p_{\slot'}$, for each $1 \leq \slot \leq \qtdSlot$;
    (iii) $vl_{\ag} = vl_{\ag}'$. 
\end{itemize}

Notice there is exactly one initial state for each possible valuation profile in $(\prod_{\ag\in \Ag}\Valuation_\ag)$. Additionally,  valuations remain unchanged after the initial state. 
We use the formula $\slot^{-1} \egdef 1\backslash \slot$ when it is convenient to  obtain a value in $[-1,1]$ for representing a slot $\slot$. 
The utility of agent $\ag$ when she is assigned to slot $\slot$ is denoted by the formula  
$\phiutil_{\ag, \slot} \egdef 
\theta_\slot \times (\valprop_\ag - \paymentprop_{\slot} )$.
The expected utility for agent $\ag$ depends on her actual allocation, that is, 
$
\phiutil_{\ag} \egdef 
\sum_{\slot \in \Slot} \allocprop_{\ag, \slot} \times \phiutil_{\ag, \slot}
$.

\subsection{Solution concepts for $\wCGS_{GSP}$} 
\label{sec:properties} 
In this section, we show how \NatSLF\ can be used for the 
verification of mechanisms with natural strategies.  
In sight of our motivating example, we aim at rephrasing conditions and properties usually considered in the analysis of keyword auctions 
\cite{Cary2007,edelman2007internet,
varian2007position}. 
 
\head{Nash equilibrium}
Since auctions are noncooperative, the solution concept in the pure strategy setting usually considered is the Nash equilibrium (\NE).  The \NE\ captures the notion of stable solution: a strategy profile is \NE\ if no player can improve her utility through an unilateral change of strategy \cite{AGTbook}.  
With \NatSLK, we restrict the range of strategies to  simple ones, 
as it enables us to reason about artificial agents with limited capabilities and human-friendly strategies.  
Let  $\profile\strat=(\strat_\ag)_{\ag\in\Ag}$ be a profile of strategies 
and $k>0$ and define the formula 
\begin{align*}
\NE(\profile \strat, k) \egdef \bigwedge_{\ag \in \Ag}\Astrat[\varb]  \big[(\Ag_{-\ag},\strat_{-\ag}) & (\ag, \varb) \X \phiutil_{\ag}
\pref 
(\Ag,\profile\strat)\X \phiutil_{\ag} \big]
\end{align*}

The formula $\NE(\profile \strat, k)$ means that, for every agent and alternative strategy $\varb$ of complexity at most $k$, binding to $\varb$ when everyone else binds to their strategies in $\profile \strat$ leads to at most the same utility as when she also binds to her strategy in $\profile \strat$. 
In relation to strategies with complexity at most $k$, the strategy profile $\profile\strat$ leads to a \NE\
in the next state of $\pos$ 
if 
$\semSL{\CGS}{\setting}{\assign}{\pos}{ \NE(\profile \strat, k)} = 1$. 

Predicting outcomes of a keyword auction is a difficult task given the infinite nature of NE continuum \cite{Yuan2017}. For this reason, refined solution concepts have been proposed to reduce the \NE\ continuum to subsets. 
Edelman \textit{et al. } \cite{edelman2007internet} studied the subset called locally envy-free equilibrium (\LEFE), in which no advertiser can improve her utility by exchanging her current slot to the one ranked one position above, given the current prices.  

\head{Locally envy free equilibrium}
Let  $\profile\strat=(\strat_\ag)_{\ag\in\Ag}$ be a profile of strategies, we define the formula 
\begin{align*}
\LEFE(\profile \strat) \egdef \bigwedge_{\ag \in \Ag} 
(\Ag,\profile \strat)
\X \big[
\LEF^\ag \land \LEFloser^\ag
]
\end{align*}
where $\LEF^\ag \egdef \bigwedge_{1<\slot\leq\qtdSlot}  ( \allocprop_{\ag,\slot} = 1 \to \phiutil_{\ag, \slot} \geq \phiutil_{\ag, \slot-1} )$ indicates that when an agent is allocated to a slot, she does not prefer to switch to the slot right above and $\LEFloser^\ag  \egdef ( \bigwedge_{\slot\in\Slot} \allocprop_{\ag,\slot} = 0) \to 0 \geq \phiutil_{\ag, \qtdSlot}$ denotes that agents who were not assigned to any slot do not prefer to get the last slot. 

$\LEFE(\profile \strat)$ means that, for any agent, when everyone follows the strategies in $\profile \strat$, it holds  that (i) if she wins  $\slot$, her utility for $\slot$ is greater than for slot $\slot-1$ (at current prices) and (ii) if she does not get any slot, then her utility for the last slot is at most zero. 
Strategy profile $\profile\strat$ leads to \LEFE\
in the next state of $\pos$ if  $\semSL{\CGS}{\setting}{\assign}{\pos}{ \LEFE(\profile \strat)} = 1$. 
Based on \cite{edelman2007internet,varian2007position}, 
we have that any \LEFE\ is also a \NE: 

\begin{proposition}
\label{prop:LEFEisNE}
For any complexity $k\geq0$, state $\pos \in \setpos$, $\setting \in \{iR,ir\}$ and strategy profile $\profile \strat = (\strat_\ag)_{\ag \in \Ag}$ with $\strat_\ag \in \setstrat^{\setting}_\ag$ for each agent  $\ag \in \Ag$, 
$\semSL{\CGS_{GSP}}{\setting}{}{\pos}{ \LEFE(\profile \strat) \to \NE(\profile \strat, k)} = 1$. 
\end{proposition}
\begin{proof}[Proof sketch] 
Let $\assign$ be an assignment, $k\geq0$ be a complexity bound for strategies, $\pos \in \setpos$ be a state, $\setting \in \{iR,ir\}$,  and $\profile \strat = (\strat_\ag)_{\ag \in \Ag}$ be a profile of $\setting$-strategies. 
Assume $\semSL{\CGS_{GSP}}{\setting}{\assign
}{\pos}{ \LEFE(\profile \strat)} = 1$.  
Let $\pos_{\profile \strat}=\trans(\pos, \profile b)$ where  $\profile b = (b_\agb)_{\agb \in \Ag}$ and $b_\agb = \action_{\match(\pos, \sigma_\agb)}$ denotes the action performed by $\agb$ in $\pos$ if she follows $\strat_\agb$. 
By definition, $\val(\pos_{\profile \strat},\allocprop_{\agb,\slot}) = 1$ iff $\agb = rank_{\profile b}(\slot)$ is the winner of slot $\slot$ and her payment is $\val(\pos_{\profile \strat},\paymentprop_\slot) = b_{rank_{\profile b}}(\slot+1)$. 
For each slot $\slot \in \{1, \dots, min(\qtdSlot, \qtdAg)\}$,  we consider whether its winner $\ag = rank_{\profile b}(\slot)$ 
could improve her utility by deviating to strategy $\varb \in \{\alpha \in \setstrata^{\setting}        : compl_{\setting}(\alpha) \leq k\}$. The case for agents who were not assigned any slot is proved similarly. 
Denote by $\bar{b} = \match(\pos, t) \in \Act$ the action that $\ag$ would take if she followed $\varb$ and $\pos_{(\strat_{-\ag}, \varb)} =\trans(\pos, (b_{-\ag}, \bar{b}))$ the next reached state from $\pos$ when she follows $\varb$ and others play according to $\strat$. 

If $\slot = 1$, any $\bar{b}\geq \val(\pos_{\profile \strat},\paymentprop_\slot)$ does not change the outcome of the auction and 
$\semSL{\CGS_{GSP}}{\setting}{\assign
}{\pos}{(\Ag_{-\ag},\strat_{-\ag}^{\setting}) (\ag, \varb) \X \phiutil_{\ag} = (\Ag,\profile\strat)\X \phiutil_{\ag}}=1$. The same holds when  $1<\slot<m$ for $\bar{b} \in [\val(\pos_{\profile \strat},\paymentprop_\slot), \val(\pos_{\profile \strat},\paymentprop_{\slot-1})]$ and when  $\slot=\qtdSlot$ for $\bar{b} \leq \val(\pos_{\profile \strat},\paymentprop_\qtdSlot)$. 
In the remaining cases, $\ag$ would change her position in $rank_{\profile b}$ with other agent. 
By the results of
\cite{edelman2007internet} (see Lemma 1), the outcome given by bids $\profile b$ is a stable assignment, that is, no advertiser can profitably rematch by changing her position with any other advertiser. Thus, $\semSL{\CGS_{GSP}}{\setting}{\assign}{\pos}{
(\Ag_{-\ag},\strat_{-\ag}) (\ag, \varb) \X \phiutil_{\ag}
\pref  
(\Ag,\profile\strat)\X \phiutil_{\ag}} =1$. 
\end{proof}

As \LEFE\ is still an equilibrium continuum, 
Edelman \textit{et al. } \cite{edelman2007internet} characterize an equilibrium in which the slot allocation and payments coincide with the ones in the dominant-strategy equilibrium (DSE) of the Vickrey–Clarke–Groves (VCG) mechanism.

Let $\profile \valuation = (\valuation_\ag)_{\ag \in \Ag}$ be a valuation profile. 
Truthfully reporting $\profile \valuation$ is the DSE of VCG \cite{AGTbook}. For each slot $\slot$ and agent $\ag$, the allocation rule for VCG in the keyword auction is the same as under GSP \cite{edelman2007internet}: 
$\VCGallocprop_{\ag,\slot}(\profile \valuation) = 1 \text{ if } rank_{\profile \valuation}(\slot)=\ag \text{ and }\slot \leq \qtdAg$. Otherwise, $\VCGallocprop_{\ag,\slot}(\profile \valuation) = 0 $. The payment for the last slot $\qtdSlot$ is 
$\VCGpaymentprop_{\qtdSlot}(\profile \valuation) = \theta_{\qtdSlot} \cdot 
     \valuation_{rank_{\profile\valuation}(\qtdSlot+1)}$ if $\qtdSlot+1\leq \qtdAg$ and $\VCGpaymentprop_{\qtdSlot}=0$ otherwise. For the remaining slots $1\leq \slot < \qtdSlot$,  
     $\VCGpaymentprop_{\slot}(\profile \valuation) = (\theta_\slot- \theta_{\slot+1} )\cdot 
     \valuation_{rank_{\profile\valuation}   (\slot+1)} + \VCGpaymentprop_{\slot+1}(\profile \valuation)$. We assume $\VCGpaymentprop_\slot(\profile\valuation)$ and $\VCGallocprop_{\slot, \ag}(\profile\valuation)$ are functions in $\Func$.

\head{VCG outcome} The following formula denotes whether the allocation and payments in the next state are the same as the ones for the VCG when agents bid truthfully:
 \[\phiVCG(\profile \strat) \egdef (\Ag,\strat) \X\big[\bigwedge_{\slot\in\Slot}( \paymentprop_\slot = \VCGpaymentprop_\slot(\profile\valprop) \land \bigwedge_{\ag \in \Ag} \allocprop_{\slot, \ag} = \VCGallocprop_{\slot, \ag}(\profile\valprop))\big]\]

If  a strategy profile leads to the VCG outcome, then 
it is a LEFE: 
\begin{proposition}
\label{prop:VCGeqisLEFE}
For any  state $\pos \in \setpos$, $\setting \in \{iR,ir\}$ and strategy profile $\profile \strat = (\strat_\ag)_{\ag \in \Ag}$ with $\strat_\ag \in \setstrat^{\setting}_\ag$ for each $\ag$, 
$\semSL{\CGS_{GSP}}{\setting}{}{\pos}{\phiVCG(\profile \strat) \to \LEFE(\profile \strat)} = 1$. 
\end{proposition}
\begin{proof} [Proof sketch] 
Let $\assign$ be an assignment, $\pos \in \setpos$ be a state, $\setting \in \{iR,ir\}$ and $\profile \strat = (\strat_\ag)_{\ag \in \Ag}$ be a profile of $\setting$-strategies. 
We denote $\profile\valuation = (\valuation_\ag)_{\ag \in \Ag}$ where  $\valuation_\ag = \val(\pos, \valprop_{\ag})$.
Assume $\semSL{\CGS_{GSP}}{\setting}{\assign}{\pos}{\phiVCG(\profile \strat)} = 1$, then we have $\semSL{\CGS_{GSP}}{\setting}{\assign}{\pos}{(\Ag,\profile\strat)\X \big[\bigwedge_{\slot \in \Slot} (\paymentprop_\slot = \VCGpaymentprop_\slot(\profile \valprop) \land \bigwedge_{\ag \in \Ag} \allocprop_{\slot, \ag} = \VCGallocprop_{\slot, \ag}(\profile \valprop))\big]}=1$. 
We denote by $\pos_{\profile \strat} = \trans(\pos, \profile b)$ the state succeeding $\pos$ when agents follow $\profile \strat$, where  $\profile b = (b_\agb)_{\agb \in \Ag}$ and $b_\agb = \action_{\match(\pos, \sigma_\agb)}$. 
Let $\slot>1$ be a slot and $\ag$ be an agent. 
By the definition of $\phiVCG$, $ \val(\pos_{\profile \strat},\allocprop_{\slot, \ag}) = 1$ if $rank_{\profile \valuation}(\slot)=\ag $ and $\slot \leq \qtdAg$. Otherwise, $\VCGallocprop_{\ag,\slot} = 0$. 
    
Given that allocations in $\pos_{\profile \strat}$ are the same as in the (truthful) outcome of VCG, it must be the case that $rank_{\profile b}(\slot) = rank_{\profile \valuation}(\slot)$. Thus, agents were allocated by descending order of their valuations (recall the valuations are distinct). 

According to the weight function, each agent is allocated to at most one slot. 
We consider first the case in which $\ag$ is not allocated to any slot, \ie\  $max_{\slot' \in \Slot}(\val(\pos_{\profile \strat},\allocprop_{\ag,\slot'})) = 0$.  This case happens when $\qtdSlot<\qtdAg$, that is, there are no enough slots for all agents. 
The utility of $\ag$ for the slot $\qtdSlot$ is $\phiutil_{\ag, \qtdSlot} =  \theta_{\qtdSlot}  (\valuation_\ag - \val(\pos_{\profile \strat},\paymentprop_\qtdSlot))$.  
By the definition of $\VCGpaymentprop_\qtdSlot$, 
$\phiutil_{\ag, \qtdSlot} =  \theta_{\qtdSlot}  (\valuation_\ag -  \theta_{\qtdSlot}        \valuation_{rank_{\profile\valuation(\qtdSlot+1)}}     )$. That is,  $\phiutil_{\ag, \qtdSlot} =  \theta_{\qtdSlot}  \valuation_\ag -  \theta_{\qtdSlot}^2        \valuation_{rank_{\profile\valuation(\qtdSlot+1)}}$. Since  $\valuation_{\ag}<\valuation_{rank_{\profile b}(\qtdSlot+1)}$, we have that $\phiutil_{\ag, \qtdSlot}< 0$. Thus, $\semSL{\CGS_{GSP}}{\setting}{\assign}{\pos_{\profile \strat}}{ 0 > \phiutil_{\ag, \qtdSlot}} = 1$ and $  \semSL{\CGS_{GSP}}{\setting}{\assign}{\pos_{\profile \strat}}{\LEFloser}=1$. 

Now we verify the case $\ag$ was assigned to a  slot $1<\slot\leq\qtdSlot$. Assume for the sake of contradiction, that $  \semSL{\CGS_{GSP}}{\setting}{\assign}{\pos_{\profile \strat}}{ \phiutil_{\ag, \slot-1} > \phiutil_{\ag, \slot}} = 1$. 
Then, in the game induced by VCG, $\ag$ would have an incentive to switch her bid with the agent in slot $\slot-1$, which is a contradiction since the bidding $\valuation_\ag$ is the dominant strategy for $\ag$ in VCG. 
\end{proof}

In fact, from \cite{edelman2007internet,varian2007position} the VCG payments are the lower bound of locally envy-free equilibrium. Thus, in any other locally envy-free equilibrium the total revenue obtained by GSP is at least as high as the one obtained by VCG in equilibrium.
\begin{corollary}
For any  state $\pos \in \setpos$, $\setting \in \{iR,ir\}$ and strategy profile $\profile \strat = (\strat_\ag)_{\ag \in \Ag}$ with $\strat_\ag \in \setstrat^{\setting}_\ag$ for each agent $\ag$, 
$\semSL{\CGS_{GSP}}{\setting}{}{\pos}{ \LEFE(\profile \strat) \to \sum_{\slot \in \Slot}(\paymentprop_\slot) \geq  \sum_{\slot \in \Slot}(\VCGpaymentprop_\slot) }=1$. 
\end{corollary}

The solution concepts characterized in the previous section are considered 
in a single stage of the game.  
Since the auction is repeated, advertisers can change their bids very frequently and one may investigate whether the  prices stabilize and at what values  \cite{edelman2007internet}. 
Stable bids must be best responses to each other, that is, the bids form an (one-shot) equilibrium.  
Cary \textit{et al.} \cite{Cary2007} raises the problem on whether there exists 
a ``natural bidding strategy'' for the advertisers that would lead to equilibrium. 

\head{Convergence} The concept of convergence or stabilization can be  easily encoded in \NatSLF:  we say a \wCGS $\wCGS$ converge to a property $\phi$ if the initial states lead to $\phi$ being eventually always the case. Formally, a \wCGS converge to a condition $\phi$ if $\semSL{\CGS}{\setting}{\assign}{\pos_\init}{\F\Gg(\phi)} = 1
$ for each initial state $\pos_\init \in \setpos_\init$.

\subsection{Natural Strategies for $\wCGS_{GSP}$ }
Given agent $\ag$ and the wCGS $\wCGS_{GSP}$, we exemplify strategies for $\ag$ in a repeated keyword auction. 
For readability, we omit the epistemic operator $\Ka$ from an epistemic condition $\Ka \phi$ when the satisfaction value of $\phi$ is  known by $\ag$ in all states. 
A common approach for an advertiser is to assume that all the other bids will remain fixed in the next round and target the slot that maximizes her utility at  current prices. This mechanism allows a range of bids that will result in the same outcome from $\ag$’s perspective, so a number of 
strategies are distinguished by the bid choice  within this range.  


\head{Balanced bidding}
In the balanced bidding strategy (BB) \cite{Cary2007}, the agent 
bids so as to be indifferent between successfully winning the targeted slot at its current price, or winning a slightly more desirable slot at her bid price. 
The natural strategy representing  balanced bidding for agent $\ag$ is denoted $BB_\ag$ and is constructed in three parts. 
First, include the guarded actions  $(BB_{\ag,1}(b),b)$ for each action $b\in \Act$. Second, include $(BB_{\ag,2}(b,\slot),b)$ for each $b\in \Act$ and $1<\slot\leq\qtdSlot$. Third, the last guarded action is $(\top, 0)$.  The condition $BB_{\ag,1}(b)$ refers to the case in which the slot maximizing $\ag$'s utility is the top slot and $b$ 
is $(\valprop_\ag+\paymentprop_{1})/2$: 
\[BB_{\ag, 1}(b) \egdef 
b= \frac{\valprop_\ag+\paymentprop_{1}}{2} \land  (argmax_{\slot\in \Slot(\phiutil_{\ag, \slot})})^{-1} = 1 
\]

Condition $BB_{\ag, 2}(b, \slot)$ denotes the case in which the slot $\slot\neq1$ maximizes $\ag$'s utility and  $b$  is the bid value that is high enough to force the prices paid by her competitors to rise, but not so high that 
she would mind getting a higher slot at a price just below $b$. 
\begin{align*}
BB_{\ag, 2}(b, \slot) \egdef \phiutil_{\ag, \slot} 
= \theta_{\slot-1} \times (\valprop_\ag - b) 
\land (argmax_{\slot' \in\Slot
}
(\phiutil_{\ag, \slot'} 
))^{-1} = \slot^{-1}
\end{align*} 
Notice the guarded action $BB_{\ag, 2}(b, \slot)$ is  defined for $\slot>1$ since it compares the utility with the one for $\slot-1$. The case $\slot=1$ is treated by the guarded action $BB_{\ag, 1}(b)$. 

Given a valuation profile $\profile\valuation = (\valuation_\ag)_{\ag \in \Ag}$, let $\eta_x$ be the agent in the $x$-th position of $rank_{\profile \valuation}$ (that is, $\eta_x$ is the agent with $x$-th highest valuation).
We let $b_{\eta_x}(\profile \valuation)$ be a function in $\Func$ 
defined as follows: 
\[b_{\eta_x}(\profile \valuation) = 
\begin{dcases}
     \frac{\theta_{x}}{\theta_{x-1}}
    \cdot b_{rank_{\profile \valuation}(x+1)}(\profile \valuation)+ (1- \frac{\theta_{x}}{\theta_{x-1}}
    )\valuation_{\eta_x} & \text{ if } x\geq \qtdSlot+1  \\
    \valuation_{\eta_x} & \text{ if } 2 \leq x\leq \qtdSlot
\end{dcases}
\] 

If $\profile{BB}= (BB_\ag)_{\ag \in \Ag}$ converges to the equilibrium with VCG outcomes, the agent with the highest valuation 
bids any value above $b_{\eta_2}(\profile \valuation)$. The equilibrium bid for $\ag \neq \eta_1$
is  $b_{\ag}(\profile \valuation)$ \cite{Cary2007}. 
When there are two slots and all players update their bids according to BB, the game converges to the equilibrium with VCG outcome. However, this is not the case for  more than two slots \cite{Cary2007}.

\begin{proposition}
\label{prop:bb} 
For any initial state $\pos_\init \in \setpos_\init$,  state $\pos \in \setpos$, and $1< x \leq \qtdAg$,  the following holds, where $\profile \valuation = (\val(\pos, \valprop_\ag))_{\ag \in \Ag}$:

\begin{enumerate}
    \item  \label{prop:BB0} If $\semSL{\CGS_{GSP}}{ir}{}{\pos}{\phiVCG(\profile{BB})  } = 1$, then $\action_{\match(\pos, BB_{\eta_x})} = b_{\eta_{x}}(\profile \valuation)$ and $\action_{\match(\pos, BB_{\eta_1})}>b_{\eta_2}(\profile \valuation)$; 
    
    \item \label{prop:BB1} If $\qtdSlot=2$, then   $\semSL{\CGS_{GSP}}{ir}{}{\pos_\init}{\F\Gg(\phiVCG(\profile{BB}))} = 1$;
    
    \item \label{prop:BB2} If $\qtdSlot\geq3$, then   $\semSL{\CGS_{GSP}}{ir}{}{\pos_\init}{\F\Gg(\phiVCG(\profile{BB}))} \neq 1$.    
\end{enumerate}
\end{proposition}
\begin{proof} [Proof sketch] 
Statement (\ref{prop:BB0}) is derived in \cite{Cary2007} from the results of Edelman \textit{et. al} \cite{edelman2007internet}.   
Notice that when each agent $\ag$ is bound to the natural strategy $\BB_\ag$, they will update their bids simultaneously in every state reachable form $\pos$. Thus, it corresponds to the synchronous setting described by \cite{Cary2007}. 
The proof for Statements (\ref{prop:BB1}) and (\ref{prop:BB2}) are very similar to the one provided in the analysis of the synchronous setting by Cary \textit{et. al} \cite{Cary2007}. 
\end{proof}

\head{Restricted BB}
The restricted balanced bidding strategy  (RBB) \cite{Cary2007}  is a variation of BB in which the agent only targets slots that are no better than her current slot. 
The natural strategy representing RBB for agent $\ag$ is denoted $RBB_{\ag}$ and is constructed as follows.
First, include the guarded actions $(RBB_{\ag,1}(b),b)$
for each action $b\in \Act$. Second, include $(RBB_{\ag,2}(b,\slot),b)$ for each $b\in \Act$ and $1<\slot\leq\qtdSlot$. Finally, the last guarded action is $(\top, 0)$.  
Let $\slot_\ag = min(\qtdSlot, \sum_{\slot' \in \Slot} \slot'\times \allocprop_{\ag, \slot'})$ be the slot assigned to agent $\ag$ or the last slot if there is no such slot. 
Define $RBB_{\ag,1}(b)$ and $RBB_{\ag,2}(b, \slot)$: 
\[RBB_{\ag, 1}(b) \egdef 
b= \frac{\valprop_\ag+\paymentprop_{1}}{2} 
\land  
argmax_{\slot\in \Slot\& \slot \geq \slot_\ag}(
\phiutil_{\ag, \slot} 
)=1 
\] 
\begin{align*}
RBB_{\ag, 2}(b,\slot) \egdef   
\phiutil_{\ag, \slot}
= \theta_{\slot-1} \times (\valprop_\ag - b) 
\\
\land (argmax_{\slot' \in\Slot\& \slot' \geq \slot_\ag}(\phiutil_{\ag, \slot'}
))^{-1} = \slot^{-1}    
\end{align*}

Similar to the results in 
\cite{Cary2007},
we have that if all agents follow the  restricted balanced-bidding strategy, the auction converge to the VCG equilibrium outcome.  RBB always converge:

\begin{proposition}
\label{prop:rbb}
 For any initial state $\pos_\init \in \setpos_\init$,  state $\pos \in \setpos$, and $1< x \leq \qtdAg$, 
 the following holds, where $\profile \valuation = (\val(\pos, \valprop_\ag))_{\ag \in \Ag}$:
\begin{enumerate}
    \item \label{prop:rbb0} If $\semSL{\CGS_{GSP}}{ir}{}{\pos}{\phiVCG(\profile{RBB})  } = 1$, then $\action_{\match(\pos, RBB_{\eta_x})} = b_{\eta_{x}}(\profile \valuation)$ and $\action_{\match(\pos, RBB_{\eta_1})}>b_{\eta_2}(\profile \valuation)$;
    
    \item  \label{prop:rbb1} $\semSL{\CGS_{GSP}}{ir}{}{\pos_\init}{\F\Gg(\phiVCG(\profile{RBB}))} = 1$.
\end{enumerate}
\end{proposition}
\begin{proof} [Proof sketch] 
Statement (\ref{prop:rbb0}) is a derivation from the results presented in \cite{edelman2007internet}.
For Statement (\ref{prop:rbb1}), the proof is similar to the one provided in \cite{Cary2007}. The proof idea is the following. 
First bound the number of steps until convergence of the price of slot $\qtdSlot$ and the set of players who will not be allocated slots. After this step, no losing player can afford a slot and their bids do not interfere with the convergence of the top $\qtdSlot$ agents. 
The second stage of the proof is to show that the allocation of the top $\qtdSlot$ players converges to a fixed point (in which they are sorted by their valuations). Then, for $1\leq i \leq \qtdSlot$, the proof inductively considers the allocation of slots $[i + 1, \qtdSlot]$. A subset of slots is called stable if the allocation is in order of decreasing values and if agent $\eta_j$ is the player currently allocated slot $j$, then her last bid is in accordance with $b_{\eta_j}(\profile \valprop)$ for every $j \in [i+1, \qtdSlot]$. While the current setting is not a fixed point of RBB, the proof proceeds by characterizing the number of rounds taken for increasing the size of the maximal stable set. 
\end{proof}

\head{Knowledge grounded RBB} 
The knowledge grounded RBB strategy (KBB) is a variation of RBB in which the agent uses her knowledge about the valuation of the player currently at her target slot to ground her bid value. The idea is to avoid bidding  more than what she knows her opponent valuates the slot. 
The natural strategy representing KBB for agent $\ag$ is denoted $KBB_\ag$ is constructed in three steps. First, include the guarded actions $(KBB_{\ag, 1}(b,c, \agb),c)$ for each $b,c\in \Act$ and agent $\agb\neq \ag$. Second, include $(KBB_{\ag, 2}(b,\slot,c, \agb),c)$ for each $b,c\in \Act$, slot $1<\slot \leq \qtdSlot$ and agent $\agb\neq \ag$. Finally, include the guarded actions from $RBB_\ag$. The conditions $KBB_{\ag, 1}(b,c, \agb)$ and $KBB_{\ag, 2}(b,\slot,c, \agb)$  are defined as follows: 
\begin{align*}
KBB_{\ag, 1}(b,c, \agb) \egdef \Ka\big(RBB_{\ag, 1}(b) \land \allocprop_{\agb, 1} = 1 \land c = min(\valprop_{\agb}, b) \big)
\end{align*} 
\begin{align*}
KBB_{\ag, 2}(b,\slot,c, \agb) \egdef \Ka\big(RBB_{\ag, 2}(b, \slot) \land \allocprop_{\agb, \slot} = 1 \land c = min(\valprop_{\agb}, b) \big)
\end{align*}
  
The prices 
under KBB are at most the same as under RBB:

\begin{proposition}
\label{prop:kbb}
For any state $\pos \in \setpos$, slot $\slot \in \Slot$ and agent $\ag \in \Ag$,
$\semSL{\CGS_{GSP}}{ir}{}{\pos}{
(\Ag,\profile{KBB}) \paymentprop_\slot \leq (\Ag,\profile{RBB}) \paymentprop_\slot
} = 1$. 
\end{proposition}
\begin{proof}
Consequence from the construction of $\profile{KBB}$.
\end{proof}

\begin{remark}
With natural strategies, we can easily construct an strategy in which agent $\eta_x$  plays according to $b_{\eta_x}(\profile \valuation)$ (for $1<x\leq\qtdAg$) and agent $\eta_1$  bids $b_{\eta_2}+\mathsf{inc}$ when she knows others' valuations. 
\end{remark}

\head{BB with recall}
Since BB may not converge to the VCG equilibrium outcome due to loops on the slot  allocation and prices, we construct a strategy that behaves according to BB while there is no repetition in the outcome and follows RBB otherwise.  Hereafter, we show that this strategy with recall prevents the loops that hinder the convergence of BB. 
Define the set of weighted conditions $\Psi = \{\bigwedge_{\slot \in \Slot}(\paymentprop_\slot = pr_\slot \land \bigwedge_{\ag\in \Ag} \allocprop_{\ag, \slot} = al_{\ag,\slot}) : pr_\slot \in \Act$ \& $ al_{\ag,\slot} \in \{0,1\}\}$. 
The natural strategy representing  balanced bidding with recall for agent $\ag$ is denoted $BBR_\ag$ and is constructed as follows. 
First, include the guarded actions  $(BBR_{\ag,1}(\psi,  b),b)$
for each action $b\in \Act$ and condition $\psi \in \Psi$. Second, include $(BBR_{\ag,2}(\psi,b,\slot),b)$ for each $\psi \in \Psi$, $b\in \Act$ and $1<\slot\leq\qtdSlot$. Third, include 
$(BBR_{\ag,3}(\psi,  b),b)$
for each action $b\in \Act$. Fourth, include $(BBR_{\ag,4}(b,\slot),b)$ for each $b\in \Act$ and $1<\slot\leq\qtdSlot$. Finally, the last guarded action is $(\top\iteration, 0)$. 

Now we define each guarded condition in $BBR_{\ag}$. 
If the current allocation and payments have already happen in the past, $\ag$ plays according to the restricted bidding strategy: 
\[BBR_{\ag,1}(\psi,  b) \egdef \top\iteration\concat\psi\concat\top\iteration\concat(\psi \land RBB_{\ag,1}(b))\]
\[BBR_{\ag,2}(\psi,  b, \slot) \egdef \top\iteration\concat\psi\concat\top\iteration\concat(\psi \land RBB_{\ag,2}(b, \slot))\]

If there was no repetition on the payments and slot allocation, she plays according to the balanced bidding strategy:
\[BBR_{\ag,3}(\psi,  b) \egdef \top\iteration \concat BB_{\ag,1}(b)\]
\[BBR_{\ag,4}(\psi,  b, \slot) \egdef \top\iteration \concat BB_{\ag,2}(b, \slot)\]

When all agents follow the strategy profile $\profile{BBR} = (BBR_\ag)_{\ag\in \Ag}$, the game converges to the VCG equilibrium outcome. 
\begin{proposition}
\label{prop:bbr} 
 For any initial state $\pos_\init \in \setpos_\init$,  state $\pos \in \setpos$, and $1< x \leq \qtdAg$,  
 the following holds, where $\profile \valuation = (\val(\pos, \valprop_\ag))_{\ag \in \Ag}$:

\begin{enumerate}
    \item \label{prop:bbr-pt1} If $\semSL{\CGS_{GSP}}{iR}{}{\pos}{\phiVCG(\profile{BBR})  } = 1$, then  $\action_{\match(\pos, BBR_{\eta_x})} = b_{\eta_{x}}(\profile \valuation)$ and $\action_{\match(\pos, BBR_{\eta_1})}>b_{\eta_2}(\profile \valuation)$;
    
    \item  \label{prop:bbr-pt2} $\semSL{\CGS_{GSP}}{iR}{}{\pos_\init}{\F\Gg(\phiVCG(\profile{BBR})) } = 1$.
 \end{enumerate}
\end{proposition}
\begin{proof}[Proof sketch] 
Statement (\ref{prop:bbr-pt1}) follows from Propositions \ref{prop:bb} and \ref{prop:rbb}.
Statement (\ref{prop:bbr-pt2}) is proven by contradiction. Assume the game does not converge to the VCG equilibrium outcome. 
Let $\assign$ be an assignment and $\iplay=\out(\pos_\init,\assign)$ be the play starting in $\pos_\init$ and follows the strategies assigned by $\assign$.
Since $\wCGS_{GSP}$ has finitely many states, there exist two indices $g<l$ such that $\iplay_g=\iplay_l$. Thus, for every $\psi \in \Psi$, $\semSL{\CGS_{GSP}}{iR}{}{\iplay_g}{\psi} = \semSL{\CGS_{GSP}}{iR}{}{\iplay_l}{\psi}$. Then, the game proceeds according to RBB strategy. That is, for any index $j\geq l$ and agent $\ag$, 
$\action_{\match(\iplay_l, BBR_{\ag})} = \action_{\match(\iplay_l, RBB_{\ag})}$. From Proposition \ref{prop:rbb}, it follows that  $\semSL{\CGS_{GSP}}{iR}{\assign}{\pos_\init}{\F\Gg(\phiVCG(\profile{BBR})) } = 1$.
\end{proof}

When other agents are inactive (\ie\ they repeat their last action), if $BBR_\ag$ selects a different bid from the one assigned by $RBB_\ag$, the utility of $\ag$ in the next state is greater under $BBR_\ag$.  
\begin{proposition}
\label{prop:bbrutil}
Let $\setting \in \{ir, iR\}$ and $\pos = \trans(\pos', \profile \act)$, for some state $\pos' \in \setpos$ and action profile $\profile \act = (\act)_{\ag \in \Ag}$. 
Given an agent $\ag \in \Ag$,  let ${\strat^{\setting}_{-\ag}} = (\strat_{\agb}^{\setting})_{\agb \in \Ag_{-\ag}}$ be a $\setting$-strategy profile, where the strategy $\strat_{\agb}^{\setting}$ of agent $\agb$ is such that $\action_{\match(\pos, \strat_{\agb}^{\setting})} = \act_\agb$. 
If  $\action_{\match(\pos, BBR_{\ag})}\neq \action_{\match(\pos, RBB_{\ag})}$, then $\semSL{\CGS_{GSP}}{iR}{}{\pos}{(\Ag_{-\ag}, {\strat^{iR}_{-\ag}})\allowbreak(\ag,\allowbreak BBR_\ag) \X \phiutil_\ag}$ $ > $ $\semSL{\CGS_{GSP}}{ir}{}{\pos}{(\Ag_{-\ag},\allowbreak {\strat^{ir}_{-\ag}})\allowbreak (\ag,\allowbreak RBB_\ag)\allowbreak \X \phiutil_\ag}$.
\end{proposition}
\begin{proof}[Proof sketch] 
Assume the actions assigned by $BBR_\ag$ and $RBB_\ag$ are different, that is,  $\action_{\match(\pos, BBR_{\ag})}\neq \action_{\match(\pos, RBB_{\ag})}$, 
then it must be the case that $\action_{\match(\pos, BBR_{\ag})} = b$ such that $\semSL{\CGS_{GSP}}{ir}{\assign}{\pos}{
\phiutil_{\ag, \slot} 
= \theta_{\slot-1} \times (\valprop_\ag - b) 
\land (argmax_{\slot' \in \Slot
}
(\phiutil_{\ag, \slot'} 
))^{-1} = \slot^{-1}} =1   
$ for some slot $\slot>1$.  
Notice that $RBB_\ag$ selected the action $\action_{\match(\pos, RBB_{\ag})}$ that maximizes the utility among slots that are better or equal to $\ag$'s current slot. By the other hand, the condition followed by $BBR_\ag$ in $\pos$ chose the action $b$ that maximizes among all slots. 
Thus, since the actions are different, $b$ has the greatest estimated utility assuming the others repeat their previous bids. As it is in fact the case, by the definition of $\strat_{-\ag}^{iR}$ and $\strat_{-\ag}^{ir}$, $\ag$'s is assigned to $\slot$ in the next state and her utility is the one estimated, that is  $\semSL{\CGS_{GSP}}{iR}{\assign}{\pos}{\X\phiutil_\ag}=\semSL{\CGS_{GSP}}{iR}{\assign}{\pos}{\phiutil_{\ag, \slot}}$. Thus,   $\semSL{\CGS_{GSP}}{iR}{}{\pos}{(\Ag_{-\ag}, \strat_{-\ag}^{iR})(\ag,\allowbreak BBR_\ag) \X \phiutil_\ag}$ $ > $ $\semSL{\CGS_{GSP}}{ir}{\assign}{\pos}{(\Ag_{-\ag}, \strat_{-\ag}^{ir})\allowbreak(\ag, RBB_\ag) \X \phiutil_\ag}$. 
\end{proof}

\begin{remark}
In vindictive bidding  
\cite{Zhou2007}, the agent bids as high as possible to raise the payment of the advisor in the slot right below hers. Since there is the risk that a change in other agents’ bids could result in paying a higher price than expected, the player could use memory to balance the use of aggressive bids. 
\end{remark}



\newcommand{\denotation}[2][]{{\llbracket #2\rrbracket}^{#1}}
\newcommand{\set}[1]{\{#1\}}
\newcommand{\lexpr}{\preceq_e}
\newcommand{\ldist}{\preceq_d}
\newcommand{\Lsys}{\mathcal{L}}
\newcommand{\prop}[1]{\ensuremath{\mathsf{{#1}}}}
\newcommand{\onlabel}[1]{\colorbox{white}{{{#1}}}}

\section{Expressivity}\label{sec:expressivity}

In relation to \SLF with  combinatorial strategies, \NatSLF\ introduces a new, broader class of human-friendly strategies \emph{and} a language for expressing properties of agents that use such strategies.
Clearly, strategies with quantitative conditions can be used to obtain goals that would not be achievable otherwise.
On the other hand, bounded natural strategies of \NatSLF\ may not achieve some goals that can be enforced with combinatorial strategies of \SLF.
In this section, we show that the expressive power of \NatSLF\ is incomparable to that of \SLF.
In other words, there are properties of quantitative games with natural strategies that cannot be equivalently translated to properties based on combinatorial strategies, and vice versa.
From this, we conclude that reasoning about human-friendly strategies offers an inherently different view of a multi-agent system from the ``standard'' one.

\subsection{Expressive and Distinguishing Power}

We first adapt the notions of distinguishing power and expressive power to the quantitative case as follows\footnote{
  Cf.,~e.g.,~\cite{Wang09expressive} for a detailed discussion of standard notions of expressivity. }.

\begin{definition}[Distinguishing power of real-valued logics]\label{def:dist}
Let $\Lsys_1 = (L_1,\denotation{\cdot}_1)$ and $\Lsys_2 = (L_2,\denotation{\cdot}_2)$ be two logical systems with syntax $L_1,L_2$ and real-valued semantics $\denotation{\cdot}_1,\denotation{\cdot}_2$ over the same class of models $\mathcal{M}$.
We say that $\Lsys_2$ is \emph{at least as distinguishing} as $\Lsys_1$ (written: $\Lsys_1\ldist \Lsys_2$) iff for every pair of models $M,M'\in\mathcal{M}$, if there exists a formula $\phi_1\in{L}_1$ such that $\denotation[M]{\phi_1}_1 \neq \denotation[M']{\phi_1}_1$, then there is also $\phi_2\in{L}_2$ with $\denotation[M]{\phi_2}_2 \neq \denotation[M']{\phi_2}_2$.
In other words, if there is a formula of $\Lsys_1$ discerning $M$ from $M'$, then there must be also a formula of $\Lsys_2$ doing the same.
\end{definition}

\begin{definition}[Expressive power of real-valued logics]\label{def:expr}
$\Lsys_2$ is \emph{at least as expressive} as $\Lsys_1$ (written: $\Lsys_1\lexpr \Lsys_2$) iff for every $\phi_1\in{L}_1$ there exists $\phi_2\in{L}_2$ such that, for every model $M\in\mathcal{M}$, we have $\denotation[M]{\phi_1}_1 = \denotation[M]{\phi_2}_2$.
In other words, every formula of $\Lsys_1$ has a translation in $\Lsys_2$ that produces exactly the same truth values on models in $\mathcal{M}$.
\end{definition}
It is easy to see that $\Lsys_1\lexpr \Lsys_2$ implies $\Lsys_1\ldist \Lsys_2$.
Thus, by transposition, we also get that $\Lsys_1\not\ldist \Lsys_2$ implies $\Lsys_1\not\lexpr \Lsys_2$.

In the remainder, $\mathcal{M}$ is the class of pointed weighted games, i.e., pairs $(\CGS,\pos)$ where $\CGS$ is a \wCGS and $\pos$ is a state in $\CGS$.

\subsection{Expressivity of \NatSLF\ vs. \SLF}
\NatSLF\  and \SLF\ are based on different notions of strategic ability. The former refers to ``natural'' strategies, represented as mappings from regular expressions over atomic propositions to actions. The latter uses "combinatorial" strategies, represented by mappings from sequences of states to actions. Each natural strategy can be translated to a combinatorial one, but not vice versa. Consequently, \SLF\ can express that a given coalition has a combinatorial strategy to achieve their goal (which is not expressible in \NatSLF). On the other hand, \NatSLF\ allows expressing that a winning natural strategy does not exist (which cannot be captured in \SLF). 
Now we show that \NatSLF\ allows to express properties 
that cannot be captured in \SLF, and vice versa. 

\begin{figure}[t]\centering
\begin{tikzpicture}[>=latex,transform shape,scale=0.9]
 \input{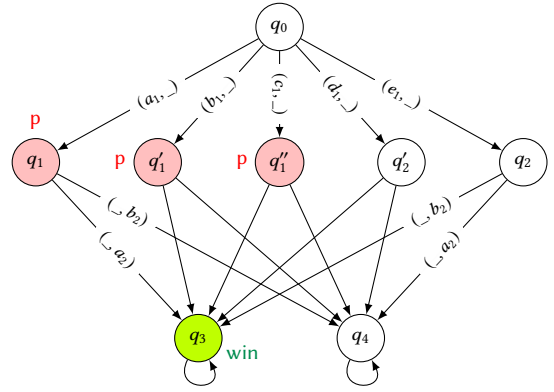}
\end{tikzpicture}
\caption{Model $\CGS_1$. Its counterpart $\CGS_1'$ is obtained by fixing $\prop{p}$ to hold only in $q_1,q_1'$. Underscore fits any action label}
\label{fig:distpower1}
\end{figure}

\begin{proposition}\label{prop:distpower1}
$\NatSLF\not\ldist\SLF$ in both $ir$ and $iR$ semantics.
\end{proposition}
\begin{proof}[Proof sketch]
Consider model $\CGS_1$ in Figure~\ref{fig:distpower1}, with agents $\Ag=\set{1,2}$, actions $\Act_1=\set{a_1,b_1,c_1,d_1,e_1}$ and $\Act_2=\set{a_2,b_2}$ available at all positions, and propositions $\APf=\set{\prop{p},\prop{win}}$.
Both propositions are qualitative (that is, the propositions have only values in \{-1,1\}).
For each proposition, the states where it evaluates to $1$ are indicated; otherwise its truth value is assumed to be $-1$.
The outgoing transitions in $q_1',q_1''$ (resp. $q_2'$) are exact copies of those at $q_1$ (resp. $q_2$).
Moreover, model $\CGS_1'$ is obtained by fixing proposition $\prop{left}$ to hold in both $q_1,q_1'$, instead of only $q_1$.
As all the propositions are qualitative, formulas of \NatSLF\ and \SLF\ evaluate to $-1$ or $1$.
Note also that the sets of $ir$ and $iR$ strategies in each model coincide, so we can concentrate on the $ir$ case w.l.o.g.

Let $\CGS\!\dagger\!\sigma$ denote the model obtained by fixing the (memoryless) strategy $\sigma$ in $\CGS$.
In order to prove that $(\CGS_1,q_0)$ and $(\CGS_1',q_0)$ satisfy the same formulas of \SLF, it suffices to observe that:
\begin{enumerate}
\item For every strategy $\sigma_1$ of agent $1$ in $\CGS_1$, there is $\sigma_1'$ in $\CGS_1'$ such that agent $2$ has the same strategic abilities in $(\CGS'_1\!\dagger\!\sigma_1,q_0)$ and $(\CGS_1'\!\dagger\!\sigma_1,q_0)$ (and vice versa).
For instance, playing $c_1$ in $\CGS$ obtains the same abilities of $1$ as playing $a_1$ in $\CGS'$.
\item Analogously for strategies of agent $2$, e.g., strategy $a_2a_2a_2b_2b_2$ in $\CGS_1$ can be simulated by strategy $a_2a_2b_2b_2b_2$ in $\CGS_1'$.
\end{enumerate}

On the other hand, the formula $\Enatstr{2}{2}\Anatstr{1}{1}(1,s_1)(2,s_2) \F\prop{win}$ of \NatSLF\ holds in $(\CGS_1,q_0)$, but not in $(\CGS_1',q_0)$.
The winning natural strategy for agent $2$ in $\CGS_1$ is $\big( (\top^*\prop{p},a_2),\ (\top^*,b_2)\big)$; clearly, it does not succeed in $\CGS_1'$.
\end{proof}

\begin{figure}[t]\centering
\begin{tikzpicture}[>=latex,scale=1]
 \input{NatSLF/Models/model-robber-perf.tex}
\end{tikzpicture}
  \qquad\qquad
\begin{tikzpicture}[>=latex,scale=1]
 \input{NatSLF/Models/model-robber-imperf.tex}
\end{tikzpicture}
\caption{Models $\CGS_2$ (left) and $\CGS_2'$ (right)}
\label{fig:distpower2}
\end{figure}

\begin{proposition}\label{prop:distpower2} 
$\SLF\not\ldist\NatSLF$ in both $ir$ and $iR$ semantics.
\end{proposition}
\begin{proof}[Proof sketch]
Consider models $\CGS_2$ and $\CGS_2'$ in Figure~\ref{fig:distpower2}.
They have isomorphic action/transition structures, the only difference being the indistinguishability of states $q_1,q_2$ in $\CGS_2'$ (but not in $\CGS_2$).
Since the two states have the same valuations of propositions, each natural strategy must specify the same decision in $q_1,q_2$.
Thus, both players have exactly the same available natural strategies in $\CGS_2$ and $\CGS_2'$, and hence $(\CGS_2,q_0)$ and $(\CGS_2',q_0)$ produce the same valuations of \NatSLF\ formulas.

On the other hand, we have that $\Estr{2}\Astr{1}(1,s_1)(2,s_2) \F\prop{win}$ of \SLF\ holds in $(\CGS_2,q_0)$, but not in $(\CGS_2',q_0)$.
\end{proof}

The following is an immediate consequence.

\begin{theorem}\label{prop:distpower}
$\NatSLF$ and $\SLF$ have incomparable distinguishing power over the class of pointed \wCGS (in both $ir$ and $iR$ semantics).
\end{theorem}

\begin{corollary}\label{prop:exprpower}
$\NatSLF$ and $\SLF$ have incomparable expressive power over the class of pointed \wCGS (in both $ir$ and $iR$ semantics).
\end{corollary}

\balance
\section{Model Checking}
\label{sec:modelchecking}

In this section we show that the model checking problem for \NatSLF\ with imperfect information is no harder than model checking LTL or classic SL with memoryless agents. 
First of all, we define 
the quantitative model-checking problem for \NatSLF.

\begin{definition}
Given $\setting \in \{ir, iR\}$, 
the model-checking problem for \NatSLF\ consists in deciding, for a given sentence $\varphi$, wCGS $\wCGS$, state $\pos \in \setpos$ 
and predicate $P \subseteq [-1; 1]$, whether $\semSL{\wCGS}{\setting}{}{\pos}{\varphi} 
\in P$.
\end{definition}

Now, we have all the ingredients to prove the following result.

\begin{theorem}
	\label{theo-mc-recall}
	Assuming that functions in $\Func$ can be computed in polynomial space, model checking \NatSLF\ with imperfect information, natural strategies with recall, and $k$ as parameter of the problem is \PSPACE-complete.
\end{theorem}

\begin{proof}
	For the lower-bound we recall that $\LTLF$ model checking is \PSPACE-complete \cite{ABK16}.
	For the upper-bound, to verify that a given \NatSLF\ formula $\varphi$ is satisfied over a \wCGS $\wCGS$ at a state $\pos \in \setpos$ 
	under assignments $\assign$ over uniform natural strategies with recall, we make use of a recursive function as is done in \cite{CLMM18}.
	We start by showing that each recursive call only needs at most polynomial space. 
	First, observe that each assignment $\assign$ has a strategy $s_a$ for each agent $a \in Ag$ \footnote{Note that, as defined in Section \ref{sec:natslf}, we consider only complete assignments. Thus, we can assume that a strategy is assigned for each agent.}. We know that each strategy $s_a$ that can be assigned to agent $a$ is bounded, and we have that $compl(s_a) \leq k$. Thus, each strategy can be stored in $O(k \cdot |Act|)$ and, by consequence, any assignment can be stored in space $O( (|Ag| \cdot |\free(\varphi)|) \cdot (k \cdot |Act|) ) $.
	Now, we can analyse the recursive function. For the base case, $\semSL{\CGS}{\setting}{\assign}{\pos}{p}$ can be computed in constant space via the weight function. For strategy quantification $\semSL{\CGS}{\setting}{\assign}{\pos}{\Estrata{\ag}\varphi}$, besides the recursive call to $\semSL{\CGS}{\setting}{\assign[\var\mapsto\strat]}{\pos}{\varphi}$ we need space $O(|k|\cdot|\Act|)$ to store the current strategy and the current maximum value computed. 
	For $\semSL{\CGS}{\setting}{\assign}{\pos}{f(\varphi_{1}, \ldots, \varphi_{m})}$, by assumption $f$ is computed in polynomial space. 
	For $\semSL{\CGS}{\setting}{\assign}{\pos}{\X\phi}$, we only need to observe that the next state in $\out(\pos,\assign)$ is computed in constant space. 
	Finally, we detail how $\semSL{\CGS}{\setting}{\assign}{\pos}{\phi_1\until\phi_2}$ is computed. Let $\iplay=\out(\pos,\assign)$. Since $\wCGS$ has finitely many states, there exist two indices $g<l$ such that $\iplay_g=\iplay_l$, and since strategies are bounded by $k$, the suffix of $\iplay$ starting at index $l$ is equal to the suffix starting at index $g$. So there exist $\fplay_1=\pos_0\ldots\pos_{g-1}$ and $\fplay_2=\pos_{g}\ldots\pos_{l-1}$ such that $\iplay=\fplay_1\cdot \fplay_2^\omega$. It follows that 
	\begin{align*}
		\semSL{\CGS}{\setting}{\assign}{\pos}{\phi_1\Uu\phi_2} &= \sup_{i \ge 0}     \min 
		\Big
		(\semSL{\CGS}{\setting}{\assign}{\iplay_{i}
		}{\phi_2}, 
		\min_{0\leq j <i} 
		\semSL{\CGS}{\setting}{\assign}{\iplay_{j}
		}{\phi_1}\Big)\\
		&= \max_{0\le i <l}    \min \Big
		(\semSL{\CGS}{\setting}{\assign}{\iplay_{i}
		}{\phi_2}, 
		\min_{0\leq j <i} 
		\semSL{\CGS}{\setting}{\assign}{\iplay_{j}
		}{\phi_1}\Big)
	\end{align*}
	This can be computed by a while loop that increases $i$, computes 
	$\semSL{\CGS}{\setting}{\assign}{\iplay_{i}}{\phi_2}$ and $\min_{0\leq j <i} \semSL{\CGS}{\setting}{\assign}{\iplay_{j}}{\phi_1}$, their minimum, and records the result if it is bigger than the previous maximum. 
	This requires to store the current value of $\min_{0\leq j <i} 
	\semSL{\CGS}{\setting}{\assign}{\iplay_{i}}{\phi_1}$, the current maximum, and the list of states already visited, which are at most $k \cdot |\setpos|$.
	Finally, the number of nested recursive calls is at most $|\phi|$, so the total space needed is bounded by $|\phi|$ times a polynomial in the size of the input, and is thus polynomial.
\end{proof}

Since memoryless natural strategies are a special case of natural strategies with recall, we obtain the following result. 

\begin{corollary}
	\label{theo-mc}
	Assuming that functions in $\Func$ can be computed in polynomial space, model checking \NatSLF\ with imperfect information, memoryless natural strategies, and $k$ as parameter of the problem is \PSPACE-complete.
\end{corollary}
\section{Conclusion}
\label{sec:conclusion}

In this work we have introduced Natural Strategy Logic with quantitative semantics and imperfect information (\NatSLF) for reasoning about 
strategic ability in auctions. 
\NatSLF\ provides a tool for mechanism design and offers a new perspective for formal verification and design of novel mechanisms and strategies. We demonstrated the usefulness of our approach by modelling and evaluating strategies for repeated keyword auctions. 

In terms of technical results, we proved that the model checking problem for \NatSLF\ is \PSPACE-complete, that is, no harder than model checking for the much less expressive language of quantitative \LTL\ (\LTLF).
We also showed that \NatSLF\ has incomparable distinguishing and expressive power to \SLF. This means that the characterizations based on simple bounded strategies offer an inherently different view of auctions and mechanism design from characterizations using combinatorial strategies of arbitrary complexity.
Amazingly, this aspect has never been studied for natural strategies, not even for the original proposal of \NatATL~\cite{natStrategy}.

We consider several directions for future work. 
First, a probabilistic extension of Strategy Logic \cite{aminof2019probabilistic} would allow handling  
mechanisms in stochastic settings with mixed strategies.  
Another direction 
is to investigate the use of strategies with recall for learning other players' valuations based on their behaviour. 
Finally, the implementation of a model checker for \NatSLF\ will enable the empirical evaluation of auctions with natural strategies. 

\begin{acks}
This research is supported by the ANR project AGAPE ANR-18-CE23-0013. 
\end{acks}

\newpage

\bibliographystyle{apalike}


\begin{thebibliography}{}

\bibitem[{\AA}gotnes, 2006]{Agotnes06action}
{\AA}gotnes, T. (2006).
\newblock Action and knowledge in alternating-time temporal logic.
\newblock {\em Synthese}, 149(2):377--409.

\bibitem[{\AA}gotnes and Walther, 2009]{Agotnes09bounded}
{\AA}gotnes, T. and Walther, D. (2009).
\newblock A logic of strategic ability under bounded memory.
\newblock {\em Journal of Logic, Language and Information}, 18(1):55--77.

\bibitem[Alechina et~al., 2007]{Alechina07aprograms}
Alechina, N., Dastani, M., Logan, B., and Meyer, J.-J.~C. (2007).
\newblock A logic of agent programs.
\newblock In {\em Proceedings of AAAI}, pages 795--800.

\bibitem[Alechina et~al., 2008]{Alechina08execution}
Alechina, N., Logan, B., Dastani, M., and Meyer, J.-J.~C. (2008).
\newblock Reasoning about agent execution strategies.
\newblock In {\em Proceedings of International Joint Conference on Autonomous
  Agents and Multiagent Systems (AAMAS)}, pages 1455--1458.

\bibitem[Alechina et~al., 2009]{Alechina09bounded}
Alechina, N., Logan, B., Nga, N., and Rakib, A. (2009).
\newblock A logic for coalitions with bounded resources.
\newblock In {\em Proc. of International Joint Conference on Artificial
  Intelligence (IJCAI)}, pages 659--664.

\bibitem[Alechina et~al., 2010]{Alechina10atl-bounded}
Alechina, N., Logan, B., Nguyen, H., and Rakib, A. (2010).
\newblock Resource-bounded alternating-time temporal logic.
\newblock In {\em Proceedings of International Joint Conference on Autonomous
  Agents and Multiagent Systems (AAMAS)}, pages 481--488.

\bibitem[Almagor et~al., 2016]{ABK16}
Almagor, S., Boker, U., and Kupferman, O. (2016).
\newblock Formally reasoning about quality.
\newblock {\em Journal of the ACM}, 63(3):24:1--24:56.

\bibitem[Alur et~al., 2002]{alur2002alternating}
Alur, R., Henzinger, T.~A., and Kupferman, O. (2002).
\newblock Alternating-time temporal logic.
\newblock {\em Journal of the {ACM}}, 49(5):672--713.

\bibitem[Aminof et~al., 2019]{aminof2019probabilistic}
Aminof, B., Kwiatkowska, M., Maubert, B., Murano, A., and Rubin, S. (2019).
\newblock Probabilistic strategy logic.
\newblock In Kraus, S., editor, {\em Proceedings of the Twenty-Eighth
  International Joint Conference on Artificial Intelligence, {IJCAI}}.

\bibitem[Baier and Katoen, 2008]{principlesmodelchecking}
Baier, C. and Katoen, J.-P. (2008).
\newblock {\em Principles of Model Checking (Representation and Mind Series)}.
\newblock The MIT Press.

\bibitem[Barlo et~al., 2008]{BCS08}
Barlo, M., Carmona, G., and Sabourian, H. (2008).
\newblock Bounded memory with finite action spaces.
\newblock {\em Sabanci University, Universidade Nova de Lisboa and University
  of Cambridge}.

\bibitem[Barthe et~al., 2016]{barthe2016computer}
Barthe, G., Gaboardi, M., Arias, E., Hsu, J., Roth, A., and Strub, P.-Y.
  (2016).
\newblock Computer-aided verification for mechanism design.
\newblock In {\em Conference on Web and Internet Economics {(WINE)}}.

\bibitem[Belardinelli et~al., 2017]{BelardinelliCDJ17}
Belardinelli, F., Condurache, R., Dima, C., Jamroga, W., and Jones, A.~V.
  (2017).
\newblock Bisimulations for verifying strategic abilities with an application
  to threeballot.
\newblock In {\em {Proc. of AAMAS 17}}, pages 1286--1295.

\bibitem[Belardinelli et~al., 2018]{BLM18}
Belardinelli, F., Lomuscio, A., and Malvone, V. (2018).
\newblock Approximating perfect recall when model checking strategic abilities.
\newblock In {\em Proc. of KR}, pages 435--444.

\bibitem[Belardinelli et~al., 2014]{BelardinelliLomuscioPatrizi15}
Belardinelli, F., Lomuscio, A., and Patrizi, F. (2014).
\newblock Verification of agent-based artifact systems.
\newblock {\em Journal of Artificial Intelligence Research}, 51:333--376.

\bibitem[Bordini et~al., 2006]{Bordini06verifying}
Bordini, R., Fisher, M., Visser, W., and Wooldridge, M. (2006).
\newblock Verifying multi-agent programs by model checking.
\newblock {\em Autonomous Agents and Multi-Agent Systems}, 12(2):239--256.

\bibitem[Bourne, 1970]{Bourne70concepts}
Bourne, L.~E. (1970).
\newblock Knowing and using concepts.
\newblock {\em Psychol. Rev.}, 77:546--556.

\bibitem[Bouyer et~al., 2019]{Bouyer19}
Bouyer, P., Kupferman, O., Markey, N., Maubert, B., Murano, A., and Perelli, G.
  (2019).
\newblock {Reasoning about Quality and Fuzziness of Strategic Behaviours}.
\newblock In {\em Proc. of {IJCAI}}.

\bibitem[Bulling and Farwer, 2010a]{Bulling10rtl}
Bulling, N. and Farwer, B. (2010a).
\newblock Expressing properties of resource-bounded systems: The logics {RTL*}
  and {RTL}.
\newblock In {\em Proceedings of Computational Logic in Multi-Agent Systems
  (CLIMA)}.

\bibitem[Bulling and Farwer, 2010b]{Bulling10bounded}
Bulling, N. and Farwer, B. (2010b).
\newblock On the (un-)decidability of model checking resource-bounded agents.
\newblock In {\em Proceedings of ECAI}, volume 215 of {\em Frontiers in
  Artificial Intelligence and Applications}, pages 567--572. IOS Press.

\bibitem[Bulling and Jamroga, 2014]{BJ14}
Bulling, N. and Jamroga, W. (2014).
\newblock Comparing variants of strategic ability: how uncertainty and memory
  influence general properties of games.
\newblock {\em Journal of Autonomous Agents and Multi-Agent Systems},
  28(3):474--518.

\bibitem[Caminati et~al., 2015]{Caminati2015}
Caminati, M., Kerber, M., Lange, C., and Rowat, C. (2015).
\newblock Sound auction specification and implementation.
\newblock In {\em ACM Conference on Economics and Computation (EC)}.

\bibitem[Cary et~al., 2007]{Cary2007}
Cary, M., Das, A., Edelman, B., Giotis, I., Heimerl, K., Karlin, A.~R.,
  Mathieu, C., and Schwarz, M. (2007).
\newblock {Greedy bidding strategies for keyword auctions}.
\newblock {\em EC 2007 - Proceedings of the Eighth Annual Conference on
  Electronic Commerce}, pages 262--271.

\bibitem[Cerm{\'{a}}k et~al., 2014]{Cermak+14}
Cerm{\'{a}}k, P., Lomuscio, A., Mogavero, F., and Murano, A. (2014).
\newblock {MCMAS-SLK:} {A} model checker for the verification of strategy logic
  specifications.
\newblock In {\em Proc. of the Int. Conf. on Computer Aided Verification
  (CAV)}, volume 8559 of {\em Lecture Notes in Computer Science}, pages
  525--532. Springer.

\bibitem[Cerm{\'{a}}k et~al., 2018]{CLMM18}
Cerm{\'{a}}k, P., Lomuscio, A., Mogavero, F., and Murano, A. (2018).
\newblock Practical verification of multi-agent systems against {SLK}
  specifications.
\newblock {\em Inf. Comput.}, 261:588--614.

\bibitem[Chatterjee et~al., 2010]{chatterjee2010strategy}
Chatterjee, K., Henzinger, T.~A., and Piterman, N. (2010).
\newblock Strategy logic.
\newblock {\em Information and Computation}, 208(6):677--693.

\bibitem[Dastani and Jamroga, 2010]{Dastani10programs-aamas}
Dastani, M. and Jamroga, W. (2010).
\newblock Reasoning about strategies of multi-agent programs.
\newblock In {\em Proceedings of {AAMAS}}, pages 625--632.

\bibitem[Deutsch et~al., 2009]{Deutsch+09}
Deutsch, A., Hull, R., Patrizi, F., and Vianu, V. (2009).
\newblock Automatic verification of data-centric business processes.
\newblock In {\em Proceedings of the 12th International Conference on Database
  Theory (ICDT09)}, pages 252--267. ACM.

\bibitem[Dima and Tiplea, 2011]{Dima11undecidable}
Dima, C. and Tiplea, F. (2011).
\newblock Model-checking {ATL} under imperfect information and perfect recall
  semantics is undecidable.
\newblock {\em CoRR}, abs/1102.4225.

\bibitem[Dixon et~al., 2012]{Dixon+12}
Dixon, C., Winfield, A., Fisher, M., and Zeng, C. (2012).
\newblock Towards temporal verification of swarm robotic systems.
\newblock {\em Robotics and Autonomous Systems}, 60(11):1429--1441.

\bibitem[Duijf and Broersen, 2016]{Duijf16strategies}
Duijf, H. and Broersen, J. (2016).
\newblock Representing strategies.
\newblock In {\em Proc. of Int. Workshop on Strategic Reasoning ({SR})}, pages
  15--26.

\bibitem[Edelman et~al., 2007]{edelman2007internet}
Edelman, B., Ostrovsky, M., and Schwarz, M. (2007).
\newblock Internet advertising and the generalized second-price auction:
  Selling billions of dollars worth of keywords.
\newblock {\em American economic review}, 97(1):242--259.

\bibitem[Feldman, 2000]{Feldman00conceptlearning}
Feldman, J. (2000).
\newblock Minimization of {Boolean} complexity in human concept learning.
\newblock {\em Nature}, 407:630--3.

\bibitem[Gammie and van~der Meyden, 2004]{GammieMeyden04a}
Gammie, P. and van~der Meyden, R. (2004).
\newblock {MCK}: Model checking the logic of knowledge.
\newblock In {\em Proc.\ of 16th Int. Conf. on Computer Aided Verification
  (CAV)}, volume 3114 of {\em Lecture Notes in Computer Science}, pages
  479--483. Springer.

\bibitem[Gonzalez et~al., 2015]{GonzalezGriesmayerLomuscio15}
Gonzalez, P., Griesmayer, A., and Lomuscio, A. (2015).
\newblock Verification of {G}{S}{M}-based artifact-centric systems by predicate
  abstraction.
\newblock In {\em Proceedings of the 13th International Conference on Service
  Oriented Computing (ICSOC15)}, volume 9435 of {\em Lecture Notes in Computer
  Science}, pages 253--268. Springer.

\bibitem[Gupta et~al., 2015]{GSW14}
Gupta, A., Schewe, S., and Wojtczak, D. (2015).
\newblock Making the best of limited memory in multi-player discounted sum
  games.
\newblock In Esparza, J. and Tronci, E., editors, {\em Proc. of the Int.
  Symposium on Games, Automata, Logics and Formal Verification (GandALF 2015)},
  volume 193 of {\em {EPTCS}}, pages 16--30.

\bibitem[Harel and Kozen, 1982]{Harel82processlogic}
Harel, D. and Kozen, D. (1982).
\newblock Process logic: Expressiveness, decidability, completeness.
\newblock {\em Journal of Computer and System Sciences}, 25(2):144--170.

\bibitem[Herzig et~al., 2014]{Herzig14explicit-progs}
Herzig, A., Lorini, E., Maffre, F., and Walther, D. (2014).
\newblock Alternating-time temporal logic with explicit programs.
\newblock In {\em Proceedings of Workshop on Logical Aspects of Multi-Agent
  Systems (LAMAS)}.

\bibitem[H{\"o}rner and Olszewski, 2009]{HO09}
H{\"o}rner, J. and Olszewski, W. (2009).
\newblock How robust is the folk theorem?
\newblock {\em The Quarterly Journal of Economics}, pages 1773--1814.

\bibitem[Jamroga and {\AA}gotnes, 2007]{JA06}
Jamroga, W. and {\AA}gotnes, T. (2007).
\newblock Constructive knowledge: what agents can achieve under imperfect
  information.
\newblock {\em J. Applied Non-Classical Logics}, 17(4):423--475.

\bibitem[Jamroga et~al., 2020]{JamrogaKM20}
Jamroga, W., Kurpiewski, D., and Malvone, V. (2020).
\newblock Natural strategic abilities in voting protocols.
\newblock In Gro{\ss}, T. and Vigan{\`{o}}, L., editors, {\em Socio-Technical
  Aspects in Security and Trust - 10th International Workshop, {STAST} 2020,
  Virtual Event, September 14, 2020}, volume 12812 of {\em Lecture Notes in
  Computer Science}, pages 45--62. Springer.

\bibitem[Jamroga et~al., 2019a]{natStrategy}
Jamroga, W., Malvone, V., and Murano, A. (2019a).
\newblock Natural strategic ability.
\newblock {\em Artificial Intelligence}, 277:103170.

\bibitem[Jamroga et~al., 2019b]{natStrategyII}
Jamroga, W., Malvone, V., and Murano, A. (2019b).
\newblock Natural strategic ability under imperfect information.
\newblock In {\em Proceedings of the International Conference on Autonomous
  Agents and MultiAgent Systems ({AAMAS})}.

\bibitem[Kerber et~al., 2016]{KERBER201626}
Kerber, M., Lange, C., and Rowat, C. (2016).
\newblock An introduction to mechanized reasoning.
\newblock {\em Journal of Mathematical Economics}, 66:26 -- 39.

\bibitem[Kouvaros and Lomuscio, 2015]{KouvarosLomuscio15}
Kouvaros, P. and Lomuscio, A. (2015).
\newblock A counter abstraction technique for the verification of robot swarms.
\newblock In {\em Proceedings of the 29th AAAI Conference on Artificial
  Intelligence (AAAI15)}, pages 2081--2088. AAAI Press.

\bibitem[Kupferman and Vardi, 2000]{KV00}
Kupferman, O. and Vardi, M.~Y. (2000).
\newblock Synthesis with incomplete informatio.
\newblock In {\em Advances in Temporal Logic}, pages 109--127. Springer.

\bibitem[Lomuscio et~al., 2015]{LomuscioQuRaimondi15}
Lomuscio, A., Qu, H., and Raimondi, F. (2015).
\newblock {MCMAS}: A model checker for the verification of multi-agent systems.
\newblock {\em Software Tools for Technology Transfer}.
\newblock http://dx.doi.org/10.1007/s10009-015-0378-x.

\bibitem[Maruyama, 2021]{maruyama2021reasoning}
Maruyama, Y. (2021).
\newblock A reasoning system for fuzzy distributed knowledge representation in
  multi-agent systems.
\newblock In {\em 2021 IEEE International Conference on Fuzzy Systems
  (FUZZ-IEEE)}, pages 1--6. IEEE.

\bibitem[Maubert et~al., 2021]{KR2021}
Maubert, B., Mittelmann, M., Murano, A., and Perrussel, L. (2021).
\newblock Strategic reasoning in automated mechanism design.
\newblock In {\em Proc.\ of the Eighteen Conference on Principles of Knowledge
  Representation and Reasoning}.

\bibitem[Mogavero et~al., 2014]{MMPV14}
Mogavero, F., Murano, A., Perelli, G., and Vardi, M. (2014).
\newblock Reasoning about strategies: On the model-checking problem.
\newblock {\em {ACM} Trans. Comput. Log.}, 15(4).

\bibitem[Montali et~al., 2014]{MontaliCalvaneseGiacomo14}
Montali, M., Calvanese, D., and {De Giacomo}, G. (2014).
\newblock Verification of data-aware commitment-based multiagent system.
\newblock In {\em Proc.\ of the 14th International Conference on Autonomous
  Agents and Multi-Agent systems (AAMAS14)}, pages 157--164. IFAAMAS.

\bibitem[Morris and Ward, 2014]{Morris14psycho-planning}
Morris, R. and Ward, G. (2014).
\newblock {\em The Cognitive Psychology of Planning}.
\newblock Psychology Press.

\bibitem[Nielsen, 1994]{Nielsen94usability}
Nielsen, J. (1994).
\newblock {\em Usability Engineering}.
\newblock Morgan Kaufmann.

\bibitem[Nov\'ak and Jamroga, 2009]{Novak09codepatterns}
Nov\'ak, P. and Jamroga, W. (2009).
\newblock Code patterns for agent oriented programming.
\newblock In {\em Proceedings of {AAMAS'09}}, pages 105--112.

\bibitem[Pauly and Parikh, 2003]{pauly2003game}
Pauly, M. and Parikh, R. (2003).
\newblock Game logic-an overview.
\newblock {\em Studia Logica}, 75(2):165--182.

\bibitem[Pauly and Wooldridge, 2003]{pauly2003logic}
Pauly, M. and Wooldridge, M. (2003).
\newblock Logic for mechanism design--a manifesto.
\newblock In {\em Proc.\ of the 2003 Workshop on Game Theory and Decision
  Theory in Agent Systems}.

\bibitem[Pnueli and Rosner, 1989]{PR89}
Pnueli, A. and Rosner, R. (1989).
\newblock {On the Synthesis of a Reactive Module.}
\newblock In {\em Symposium on the Principles of Programming Languages (POPL
  89)}, pages 179--190. ACM.

\bibitem[Reif, 1984]{Reif84}
Reif, J.~H. (1984).
\newblock The complexity of two-player games of incomplete information.
\newblock {\em J. Comput. Syst. Sci.}, 29(2):274--301.

\bibitem[Roughgarden, 2010]{AGTbook}
Roughgarden, T. (2010).
\newblock Algorithmic game theory.
\newblock {\em Communications of the ACM}, 53(7):78--86.

\bibitem[Santos, 2018]{Santos18phd}
Santos, F. (2018).
\newblock {\em Dynamics of Reputation and the Self-organization of
  Cooperation}.
\newblock PhD thesis, University of Lisbon.

\bibitem[Santos et~al., 2018]{Santos18SocialNormComplexity}
Santos, F., Santos, F., and Pacheco, J. (2018).
\newblock Social norm complexity and past reputations in the evolution of
  cooperation.
\newblock {\em Nature}, 555:242--245.

\bibitem[Troquard et~al., 2011]{THW11}
Troquard, N., van~der Hoek, W., and Wooldridge, M. (2011).
\newblock {Reasoning about Social Choice Functions}.
\newblock {\em Journal of Philosophical Logic}, 40(4):473--–498.

\bibitem[van~der Hoek et~al., 2005]{Hoek05commitment}
van~der Hoek, W., Jamroga, W., and Wooldridge, M. (2005).
\newblock A logic for strategic reasoning.
\newblock In {\em Proceedings of {AAMAS'05}}, pages 157--164.

\bibitem[Varian, 2007]{varian2007position}
Varian, H.~R. (2007).
\newblock Position auctions.
\newblock {\em international Journal of industrial Organization},
  25(6):1163--1178.

\bibitem[Vester, 2013]{Vester13ATL-finite}
Vester, S. (2013).
\newblock Alternating-time temporal logic with finite-memory strategies.
\newblock In {\em Proceedings of GandALF}, EPTCS, pages 194--207.

\bibitem[Walther et~al., 2007]{Walther07explicit}
Walther, D., van~der Hoek, W., and Wooldridge, M. (2007).
\newblock Alternating-time temporal logic with explicit strategies.
\newblock In {\em Proceedings {TARK XI}}, pages 269--278. Presses
  Universitaires de Louvain.

\bibitem[Wang and Dechesne, 2009]{Wang09expressive}
Wang, Y. and Dechesne, F. (2009).
\newblock On expressive power and class invariance.
\newblock {\em CoRR}, abs/0905.4332.

\bibitem[Wooldridge et~al., 2007]{wool2007}
Wooldridge, M., Agotnes, T., Dunne, P., and Van~der Hoek, W. (2007).
\newblock Logic for automated mechanism design-a progress report.
\newblock In {\em Proc. of AAAI}.

\bibitem[Yadav and Sardi{\~n}a, 2012]{Yadav12atl-like}
Yadav, N. and Sardi{\~n}a, S. (2012).
\newblock Reasoning about agent programs using {ATL}-like logics.
\newblock In {\em Proceedings of JELIA}, pages 437--449.

\bibitem[Yuan et~al., 2017]{Yuan2017}
Yuan, Y., Wang, F.-Y., and Zeng, D. (2017).
\newblock Competitive analysis of bidding behavior on sponsored search
  advertising markets.
\newblock {\em IEEE Transactions on Computational Social Systems},
  4(3):179--190.

\bibitem[Zhou and Lukose, 2007]{Zhou2007}
Zhou, Y. and Lukose, R. (2007).
\newblock {Vindictive bidding in keyword auctions}.
\newblock {\em ACM International Conference Proceeding Series}, 258:141--146.

\end{thebibliography}

\end{document}